\documentclass{elsarticle}
\usepackage[T1]{fontenc}
\usepackage[utf8]{inputenc}
\usepackage{amsmath,amssymb,amsthm}
\usepackage{mathtools}
\usepackage{graphicx,xcolor}
\usepackage{booktabs}
\usepackage{longtable}
\usepackage{url}
\usepackage{hyperref}
\usepackage{float}
\usepackage{placeins}
\newtheorem{theorem}{Theorem}[section]
\newtheorem{proposition}[theorem]{Proposition}

\newtheorem{corollary}[theorem]{Corollary}
\theoremstyle{definition}

\theoremstyle{remark}
\newtheorem{remark}[theorem]{Remark}
\usepackage{breqn}

\begin{document}

\begin{frontmatter}

\title{Nonlinear Reduced-Bias Extreme Value Index Estimators: A Comparative Study}

\author[1,3]{Lígia Henriques-Rodrigues\corref{cor1}%
}
\ead{lcarodrigues@ciencias.ulisboa.pt}

\author[2,4]{Frederico Caeiro}
\ead{fac@fct.unl.pt}

\author[3]{M. Ivette Gomes}
\ead{migomes@ciencias.ulisboa.pt}

\cortext[cor1]{Corresponding author}

\affiliation[1]{organization={Departamento de Ciências Matemáticas, Faculdade de Ciências, Universidade de Lisboa},
city={Lisboa},
country={Portugal}}
\affiliation[2]{organization={Departamento de Matemática, NOVA FCT and Universidade Nova de Lisboa, Faculdade de Ciências e Tecnologia},
city={Lisboa},
country={Portugal}}
\affiliation[3]{organization={CEAUL - Centro de Estatística e Aplicações, Faculdade de Ciências, Universidade de Lisboa},
city={Lisboa},
country={Portugal}}
\affiliation[4]{organization={NOVA MATH - Center for Mathematics and Applications, NOVA FCT, Universidade Nova de Lisboa},
city={Lisboa},
country={Portugal}}

\begin{abstract}
The estimation of the extreme value index (EVI) plays a central role in extreme value theory (EVT), particularly in the statistical modelling of heavy-tailed phenomena. 
 Classical semi-parametric estimators of the EVI are known to suffer from substantial bias, especially for moderate sample sizes and for common choices of the number of upper order statistics. In recent years, several reduced-bias (RB) methodologies have been proposed in order to improve finite-sample performance while preserving asymptotic efficiency. In this paper, we present new classes of nonlinear RB EVI-estimators based on power mean-type and Lehmer mean-type functionals of upper log-excesses. Under a unified higher-order framework, asymptotic representations are derived and bias and variance properties are compared relative to other RB EVI classes introduced in the literature. An asymptotic comparison at optimal levels is conducted to identify regions of the space parameter where the proposed estimators are more efficient.  A comprehensive Monte Carlo simulation study assesses  the finite sample behaviour of the  new estimators for several heavy tailed models. The results show that nonlinear reduced-bias estimators can provide improvements in bias and root mean squared error, although their relative performance depends on the underlying distribution and the choice of the tuning parameter. An application to hydrological data illustrates their practical use and the importance of threshold selection in EVI estimation.
\end{abstract}


\end{frontmatter}

\section{Introduction and theoretical framework}
Extreme value theory (EVT) provides a fundamental statistical framework for the modelling and analysis of rare and extreme phenomena (\citet{beirlant2004statistics,dehaan2006}). Its relevance is particularly evident in environmental applications, such as hydrology, climatology and air quality assessment, where accurate inference on tail behaviour is essential for risk evaluation and decision-making. Central to EVT is the extreme value index (EVI), which characterizes the tail heaviness of a distribution and determines the probabilistic structure of extreme events. 

In the context of heavy-tailed models, the EVI is commonly estimated using semi-parametric methods based on upper order statistics. Among these, the Hill estimators (\citet{hill1975}) 
remain widely used due to their simplicity and asymptotic properties. However, it is well known that classical estimators of the EVI may exhibit a pronounced bias in finite samples, strongly influenced by the choice of the number of upper order statistics involved in the estimation. This bias--variance trade-off represents one of the main practical challenges in extreme value inference.
To overcome these limitations, several reduced-bias {(RB)} approaches have been proposed in the literature. These methods typically rely on second- and higher-order tail approximations and aim to correct the leading bias term of classical estimators, while maintaining asymptotic normality and efficiency. Notable examples include bias-corrected versions of the Hill estimator (\citet{caeiro2005direct}), among others as well as {RB} extensions of more flexible classes such as the mean of order $p$--type  (\citet{brilhante2013asimple, caeiro2016mean.of.order.p}), the power mean--type (\citet{G+M-1999,G+M-2001,Segers2001, pm_red2024}) and Lehmer-type estimators (\citet{penalva2020couple,penalva2020lehmer}). The introduction of tuning parameters in these general classes allows for improved adaptability, but also raises important questions regarding their comparative performance and practical implementation.

Despite the growing number of {RB} estimators available, a unified comparative analysis of general classes of {RB} EVI-estimators remains limited, particularly under a common higher-order framework. Moreover, while theoretical properties are often well understood, comprehensive simulation-based comparisons and validations using real environmental data are still essential to assess finite-sample behaviour and stability. Also, according to {\citet{caeiro2011asymptotic}} it is not possible to find  a uniformly best class of {RB} EVI-estimators over all models and thresholds. Thus, practical performance depends on the chosen class and data scenario. 

The main objective of this paper is to address these issues by providing a systematic comparison of general classes of RB EVI-estimators. We consider linear and nonlinear {RB} corrections of power mean--type and Lehmer-type estimators within a unified second- and third-order tail framework, allowing for a coherent theoretical treatment.

\subsection{Background and theoretical framework}\label{sec:framework}


Let $(X_1,\dots,X_n)$ denote a random sample of independent and identically distributed random variables with cumulative distribution function $F$ and let 
\(
X_{1:n} \leq \cdots \leq X_{n:n}
\)
be the associated order statistics. {EVT} is concerned with the probabilistic behaviour of extreme order statistics, in particular the maximum $X_{n:n}$.

A cumulative distribution function $F$ is said to belong to the max-domain of attraction of a non-degenerate cumulative distribution function $G$, written $F \in \mathcal{D}(G)$, if there exist normalizing constants $a_n \in \mathbb{R}$ and $b_n > 0$ such that
\[
\lim_{n \to \infty} {\mathbb{P}}\left( \frac{X_{n:n}-a_n}{b_n} \leq x \right)
= \lim_{n \to \infty} F^n(a_n + b_n x) = G(x),
\qquad x \in \mathcal{C}(G),
\]
where $\mathcal{C}(G)$ denotes the continuity points of $G$.
The classical Fisher--Tippett--Gnedenko theorem {in  {\citet{fisher1928limiting, gnedenko1943}}} states that $G$ must belong to the generalized extreme value (GEV) family, given by
\[
G_\xi(x) =
\begin{cases}
	\exp\!\left\{-(1+\xi x)^{-1/\xi}\right\}, & 1+\xi x>0, \ \xi \neq 0,\\[1ex]
	\exp\!\left\{-\rm e^{-x}\right\}, & x \in \mathbb{R}, \ \xi = 0,
\end{cases}
\]
where $\xi \in \mathbb{R}$ is the {EVI, which} characterizes the max-domain of attraction of the underlying
distribution $F$. Specifically, $\xi>0$ corresponds to distributions in the Fr\'echet
max-domain of attraction, $\xi=0$ to the Gumbel max-domain of attraction, and $\xi<0$ to the Weibull max-domain of attraction, the latter being associated with a finite right endpoint. In this work, we restrict attention to the Fr\'echet case $\xi>0$, i.e., to heavy-tailed models, which is of particular
relevance in the modelling of environmental extremes.


Let $F^{\leftarrow}$ denote the generalized inverse of $F$, and define the associated quantile function
\[
U(t) := F^{\leftarrow}\!\left(1-\frac{1}{t}\right), \qquad t \geq 1.
\]
For heavy-tailed distributions, $F$ belongs to the Fr\'echet max-domain of attraction if and only if $U$ is regularly varying at infinity with index $\xi>0$, i.e. $U \in \mathcal{RV}_\xi$ (see, {\citet{bingham1987regular, bingham2007regular}}), that is,
\[
\lim_{t \to \infty} \frac{U(tx)}{U(t)} = x^{\xi}, \qquad x>0.
\]
This first-order condition (FOC){, in {\citet{deHaan1984}},} provides the theoretical basis for semi-parametric estimation of the EVI using upper order statistics. 


To study the asymptotic bias and convergence rates of EVI-estimators, the {FOC} is typically strengthened by a second-order condition (SOC). We assume that there exists a function $A(t)$, with constant sign and $|A(t)| \to 0$, {as $t\rightarrow \infty$,} such that
\begin{equation}\label{eq:SOC}
	\lim_{t \to \infty}
	\frac{\log U(tx) - \log U(t) - \xi \log x}{A(t)}
	= \frac{x^{\rho}-1}{\rho}, \qquad x>0,
\end{equation}
where $\rho \leq 0$ is a second-order parameter and $|A|$ is regularly varying at infinity with index $\rho$ (\citet{geluk1987regular}).  
The parameter $\rho$ governs the rate of convergence in the first-order approximation and plays a central role in the bias behaviour of semi-parametric EVI-estimators.


For a refined analysis of {RB} estimators, higher-order tail information is required. We therefore also consider a third-order condition (TOC). Assume that there exists a function $B(t)$, with $|B(t)| \to 0$, {as $t\rightarrow\infty$,} such that
\begin{equation}\label{eq:TOC}
	\lim_{t \to \infty}
	\frac{1}{B(t)}
	\left[
	\frac{\log U(tx) - \log U(t) - \xi \log x}{A(t)}
	- \frac{x^{\rho}-1}{\rho}
	\right]
	=
	\frac{x^{\rho+\rho_0}-1}{\rho+\rho_0},
\end{equation}
for all $x>0$, where $\rho_0 \leq 0$ is a third-order parameter and $|B|$ is regularly varying at infinity with index $\rho_0$.  
The TOC allows for a precise characterization of second-order bias terms and is essential in the derivation of asymptotic representations for {RB} estimators.\\


For technical convenience and practical relevance, we further restrict attention to the Hall--Welsh class of heavy-tailed models {in  {\citet{hall1982, hall1985adaptive}}}. In this case, the quantile function admits the expansion
\begin{equation}\label{eq:HallWelsh}
	U(t) = C t^{\xi}
	\left(
	1 + \frac{\xi \beta}{\rho} t^{\rho}
	+ \delta \beta t^{2\rho}
	+ o\!\left(t^{2\rho}\right)
	\right),
	\qquad t \to \infty,
\end{equation}
where $C>0$, $\delta, \beta \neq 0$.  
This class satisfies the first-, second- and third-order conditions given above, with 
\begin{equation}\label{eq:AB} 
A(t)=\xi \beta t^{\rho}, \qquad B(t)=\beta_0 t^{\rho_0},\qquad  \delta=\frac{\beta_0}{\beta}, \qquad \rho=\rho_0<0 . 
\end{equation}
Several distributions commonly used in environmental modelling, including Burr, Fr\'echet, generalized Pareto, log-logistic and Student-$t$ distributions, belong to Hall-Welsh class.\\

The framework introduced in this section will be used throughout the paper to derive asymptotic properties of classical and {RB} estimators of the {EVI}, and to provide a coherent basis for their theoretical and empirical comparison. \\

{The remainder of the paper is organized as follows. Section~\ref{sec:classes} reviews classical and general classes of EVI-estimators, while Section~\ref{sec:rb} presents {RB} methodologies and their asymptotic properties are presented in Section~\ref{sec:asymp} as well as an asymptotic comparison at optimal levels. Section~\ref{sec:mc} is devoted to the Monte Carlo simulation study. An environmental application is discussed in Section~\ref{sec:apps}, followed by  concluding remarks in Section~\ref{sec:conclusions}.}

\section{Classical and general classes of EVI-estimators}\label{sec:classes}

\subsection{Classical estimation: the Hill estimator}

For heavy-tailed models, the classical semi-parametric estimator of the EVI is the Hill estimator{, in {\citet{hill1975}}}, defined by
\begin{equation}\label{eq:Hill}
	\widehat{\xi}^{\rm H}(k)
	:= \frac{1}{k} \sum_{i=1}^k
	\log\!\left( \frac{X_{n-i+1:n}}{X_{n-k:n}} \right),
	\qquad 1 \leq k < n.
\end{equation}
Under the FOC and for intermediate sequences $k=k_n$, satisfying
\begin{equation}\label{intermediate}
k \to \infty \quad \text{and} \quad k/n \to 0, \qquad \text{as}\quad n \to \infty,
\end{equation}
the Hill estimator is consistent for $\xi$ {({\citet{mason1982}})}. When the SOC in \eqref{eq:SOC} holds, its asymptotic
distribution admits the representation
\begin{equation}\label{eq:HillAR}
	\widehat{\xi}^{\rm H}(k)
	\overset{d}{=}
	\xi
	+ \frac{\xi}{\sqrt{k}} Z_k
	+ \frac{A(n/k)}{1-\rho}\left(1+o_p(1)\right),
\end{equation}
where $Z_k$ is a standard normal random variable, frequently denoted by $\mathcal{N}(0,1)$ (see also {\citet{haeuslerteugels1985, hsing1991, dehaanresnick1998}} for details on the asymptotic behaviour).  The asymptotic representation in \eqref{eq:HillAR} highlights the presence of a non-negligible bias term, driven by the second-order function $A(n/k)$, which motivates the development of
bias-reduction techniques. 

If we assume that the TOC in \eqref{eq:TOC} holds with the parametrizations in \eqref{eq:AB}, we have 
\begin{equation}\label{eq:HillAR3ord}
	\widehat{\xi}^{\rm H}(k)
	\overset{d}{=}
	\xi
	+ \frac{\sigma_{\rm H}}{\sqrt{k}} Z_k
	+  b_{\rm {H}}(\rho) A(n/k) + b'_{\rm {H}}(\rho) A^2(n/k) \left(1+o_p(1)\right),
\end{equation}
with 
\begin{eqnarray}
	\sigma^2_{\rm H} &=& \xi^2 \label{eq:var_H}\\
	b_{\rm {H}}(\rho)&=& \frac{1}{1-\rho},  \label{eq:b_H}\\
	b'_{\rm {H}}(\rho)&=& \frac{\delta}{\xi (1-2\rho)}. \label{eq:b'_H}
\end{eqnarray}	

An appealing feature of the Hill estimator is that several asymptotically equivalent forms of the EVI-estimator can be obtained via fundamentally different methodological approaches, including maximum likelihood techniques and procedures based on the mean excess function. This convergence of distinct methods towards the same estimator highlights the natural role of the Hill estimator in extreme value analysis. Further discussion and related developments can be found in {\citet{beirlant2004statistics}} and references therein. For comprehensive reviews of estimation procedures for the EVI, see {\citet{gomes2015extreme, beirlant2016review},} and {\citet{deHaanPeng1998,cabral2022}} for comparison {among} several classic EVI-estimators. 

\subsection{Generalized means-type estimators }

{\bf Power mean--type estimators:} A flexible generalization of the Hill estimator was obtained in \citet{G+M-2001} and \citet{Segers2001} by considering power means
of the log-excesses. For $p>-1$, $p\neq 0$, {\citet{dekkers1989}} introduced the $p$-moment of the log-excesses
\begin{equation}\label{eq:Mn}
M^{(p)}_{k,n}
:= \frac{1}{k} \sum_{i=1}^k
\left(
\log X_{n-i+1:n} - \log X_{n-k:n}
\right)^p.
\end{equation}
The corresponding power mean--type estimator of the EVI is given by
\begin{equation}\label{eq:PM}
	\widehat{\xi}^{\rm PM}(p,k)
	:= \frac{\left(M^{(p)}_{k,n}\right)^{1/p}}{\Gamma(p+1)^{1/p}},
	\qquad p>-1,
\end{equation}
with $\widehat{\xi}^{\rm PM}(1,k)\equiv \widehat{\xi}^{\rm H}(k)$. {The case $p=0$ is obtained by continuous extension,
\begin{equation}
\label{eq:PM0}
\widehat{\xi}^{\rm PM}(0,k)
	:= e^{\gamma}\left(\prod_{i=1}^k
\left(
\log X_{n-i+1:n} - \log X_{n-k:n}
\right)\right)^{1/k},
\end{equation}
with $\gamma$ denoting the Euler–Mascheroni constant.} The introduction of the tuning parameter $p$ provides additional flexibility and allows for improved finite-sample behaviour when compared to the Hill estimator. The following theorem extends the asymptotic distributional representation of $\widehat{\xi}^{\rm PM}(p,k)$ established by \citet{pm_red2024}, originally derived for $p>0$, to the broader range $p>-0.5$. The asymptotic behaviour for values of $p\in(-0.5,0)$ is currently under investigation and lies beyond the scope of the present work.

\begin{theorem}
Under the validity of the TOC, in \eqref{eq:TOC}, and of \eqref{eq:AB}, for intermediate $k$ sequences, i.e.\ when \eqref{intermediate} holds, the estimator {$\widehat{\xi}^{\rm PM}(p,k)$} admits an asymptotic expansion of the form given in {\citet{pm_red2024}}, 
\begin{multline}\label{eq:PMAR}
	\widehat{\xi}^{\rm{PM}}(p,k)
	\overset{d}{=}
	\xi
	+ \frac{\sigma_{\rm PM}}{\sqrt{k}} Z_k^{\rm PM}
	+ b_{\rm {PM}}(\rho)\,A(n/k) 
	+ {O_p\!\left(\frac{A(n/k)}{ \sqrt{k}}\right)} \\
	+ b'_{\rm {PM}}(\rho) A^2(n/k)(1+o_p(1)),
\end{multline}

{for $p>-0.5$}, where $Z_k^{\rm PM}$ is a standard normal random variable,  and 
{
\begin{align}
\sigma^2_{\mathrm{PM}}
&=
\begin{cases}
\dfrac{\xi^2}{p^2}
\left(
\dfrac{\Gamma(2p+1)}{\Gamma^2(p+1)}-1
\right),
& p>-\dfrac12,\quad p\neq 0,\\[2mm]
\dfrac{\xi^2\pi^2}{6},
& p=0,
\end{cases}
\label{eq:var_PM}
\end{align}
\begin{align}
b_{\mathrm{PM}}(\rho)
&=
\begin{cases}
\dfrac{1}{p\rho}
\left(
\dfrac{1}{(1-\rho)^p}-1
\right),
& p>-\dfrac12,\quad p\neq 0,\\[2mm]
-\dfrac{\log(1-\rho)}{\rho},
& p=0,
\end{cases}
\label{eq:b_PM}
\end{align}
\begin{align}
b'_{\mathrm{PM}}(\rho)
&=
\begin{cases}
\dfrac{\mathcal{B}(p,\rho,\delta)}
      {2\xi p^2\rho^2},
& p>-\dfrac12,\quad p\neq 0,\\[3mm]
\dfrac{\mathcal{C}(\rho,\delta)}
      {2\xi\rho^2},
& p=0,
\end{cases}
\label{eq:b'_PM}
\end{align}
where
\begin{align*}
\mathcal{B}(p,\rho,\delta)
={}&
\frac{p[1-(2-\delta)\rho]}
     {(1-2\rho)^p}
-\frac{2(1-p\rho)}
      {(1-\rho)^p}
-\frac{p-1}{(1-\rho)^{2p}}
-\delta p\rho+1,
\\[2mm]
\mathcal{C}(\rho,\delta)
={}&
\log^2(1-\rho)
+2(1-\rho)\log(1-\rho)
-(1-2\rho+\delta\rho)\log(1-2\rho).
\end{align*}}
\end{theorem}

As in the Hill case, the leading bias term of {$\widehat{\xi}^{\rm{PM}}(p,k)$} is of order $A(n/k)$ and depends explicitly on the tail behaviour of the underlying distribution.  \\

{\bf Lehmer mean--type estimators:}  Another class of general EVI-estimators based on Lehmer means, or $L_p$--type functionals of the log-excesses was introduced in {\citet{penalva2020couple, penalva2020lehmer}}. For $p>0$, 
\begin{equation}\label{eq:Lp}
	\widehat{\xi}^{\rm L}(p,k)
	:= \frac{M^{(p)}_{k,n}}{p\,M^{(p-1)}_{k,n}},
	\qquad 1 \leq k < n,
\end{equation}
with $\widehat{\xi}^{\rm L}(1,k)\equiv \widehat{\xi}^{\rm H}(k)$ and $M^{(p)}_{k,n}$ {defined} in {\eqref{eq:Mn}}.  The $L_p$ estimators retain the consistency properties of the Hill estimator under the FOC and for intermediate $k$ sequences, while offering additional robustness and flexibility through the parameter $p$. Under the TOC, in \eqref{eq:TOC},  and {the validity} of \eqref{eq:AB}, they admit an asymptotic representation
of the form ({\citet{cabral_tese}})
\begin{equation}\label{eq:LpAR}
	\widehat{\xi}^{\rm L}(p,k)
	\overset{d}{=}
	\xi
	+ \frac{\sigma_{\rm L}}{\sqrt{k}}
	 Z_k^{\rm L}
	+ b_{\rm L}(\rho)\,A(n/k)
	+ b'_{\rm L}(\rho) A^2(n/k)(1+o_p(1)),
\end{equation}
{for $p>0.5$}, where $Z_k^{\rm L}$ is a standard normal random variable, 
\begin{eqnarray}
\sigma^2_{\rm L} &=& \frac{\xi^2 \Gamma(2p-1)}{\Gamma^2(p)},\label{eq:var_L}\\
b_{\rm L}(\rho) &=& \frac{1}{(1-\rho)^p}, \label{eq:b_L}\\
b'_{\rm L}(\rho) &=& \frac{1}{\xi \rho}\left(\frac{1}{(1-2\rho)^{p-1}}-\frac{1-\rho(1-\rho)^{p-1}}{(1-\rho)^{2p-1}}+\frac{\delta \rho}{(1 -2\rho)^p}\right) \nonumber \\
&= &  \frac{1}{\xi \rho}\left(\frac{1-2\rho+\delta\rho}{(1-2\rho)^{p}}-\frac{1-\rho(1-\rho)^{p-1}}{(1-\rho)^{2p-1}}\right).\label{eq:b'_L}
\end{eqnarray}

As in the previous cases, the bias term of order $A(n/k)$ remains present and motivates the RB versions of the $L_p$ estimators.

\section{Bias reduction in semi-parametric EVI estimation}\label{sec:rb}

\subsection{Classical RB EVI-estimators}

As shown in Section~\ref{sec:classes}, classical and general semi-parametric estimators of the EVI typically exhibit a non-negligible bias driven by second-order tail behaviour. This bias is of order $A(n/k)$ and may significantly affect inference in finite samples, particularly when the choice of the number of upper order statistics $k$ is moderate; see, for instance, \citet{deHaanPeng1998,beirlant2004statistics}. RB  methodologies aim at correcting this dominant bias term by explicitly exploiting second-order tail information. Within the framework introduced in Section~\ref{sec:framework}, such corrections can be expressed in a unified manner through multiplicative adjustments involving the second-order parameters $(\rho,\beta)$. This approach underlies several RB proposals in EVT, including bias-corrected Hill-type estimators, as the one introduced in \citet{caeiro2005direct}, with the functional form,%
\begin{equation}\label{eq:CH}
\widehat\xi^{\, \rm RBH}(k) \equiv \widehat\xi^{\, \rm RBH}_{(\hat \rho, \hat \beta)}(k) := \xi^{\rm H}(k)\left(1- \frac{\hat \beta}{1-\hat \rho}(n/k)^{\hat \rho}\right),
\end{equation}
where $(\hat \rho, \hat\beta)$ are consistent estimators of $(\rho, \beta)$.
Consistent estimation of the second-order parameters at a suitable high level, $k_1$, plays a fundamental role in bias reduction, since it enables a reduction of bias without inflating the asymptotic variance, which typically constitutes the main trade-off in bias-reduction procedures. Although a wide variety of estimators for the second-order parameters $\rho$ and $\beta$ has been proposed in the literature, in this work we focus on the estimators introduced by {{\citet{fragaalves2003anewclass}}} and by 
{{\citet{gomes2002asymptotically}}}. These estimators exhibit good theoretical properties and have been successfully applied in the context of {RB} estimation. Alternative approaches for estimating second-order parameters can be found, for example, in
{\citet{ciuperca2010semi-parametric,gomes2010tailindex,worms2012estimation,
dewet2012asymptotically,deme2013ontheestimation,caeirogomes2015}}.

The class of estimators for the parameter $\rho$ proposed in
{\citet{fragaalves2003anewclass}} was originally parameterized by a tuning parameter
$\tau \geq 0$, and later extended to allow $\tau \in \mathbb{R}$; see {{\citet{caeirogomes2006}}}. The estimator is defined as
\begin{equation}\label{rho}
	\hat{\rho}(k;\tau)
	:=
	-\left|
	\frac{3\left\{T_n(k;\tau)-1\right\}}
	{T_n(k;\tau)-3}
	\right|,
\end{equation}
where $T_n(k;\tau)$ is constructed from the empirical moments
$M^{(p)}_{(k,n)}$ defined in \eqref{eq:Mn}. For $\tau=0$, we adopt the convention
$a^{b\tau} = b \log a$. The statistic $T_n(k;\tau)$ is given by
\begin{equation}\label{Tn}
	T_n(k;\tau)
	=
	\frac{
		\left(M_{k,n}^{(1)}\right)^\tau
		-
		\left(M_{k,n}^{(2)}/2\right)^{\tau/2}
	}{
		\left(M_{k,n}^{(2)}/2\right)^{\tau/2}
		-
		\left(M_{k,n}^{(3)}/6\right)^{\tau/3}
	},
	\qquad \tau \in \mathbb{R}.
\end{equation}

Theoretical considerations and simulation results, reported in
\citet{fragaalves2003anewclass,caeirogomes2015}, indicate that, in practical applications,
the choice $\tau=0$ is appropriate when $-1 \leq \rho < 0$, whereas $\tau=1$ is preferable
when $\rho < -1$.

For the estimation of the second-order parameter $\beta$, we consider the estimator
introduced by {{\citet{gomes2002asymptotically}}}, which is based on the
scaled log-spacings
\[
U_i := i\left( \log X_{n-i+1:n} - \log X_{n-i:n} \right),
\qquad i = 1,\ldots,k < n.
\]

Given a consistent estimator $\hat{\rho}$ of $\rho$, the corresponding estimator of
$\beta$, denoted by $\hat{\beta}(k;\hat{\rho})$, is defined for $\rho<0$
\begin{equation}\label{beta}
	\hat{\beta}(k;{\hat\rho})
	=
	\frac{
		\left(\frac{k}{n}\right)^{\hat\rho}
		\left[
		\left(\frac{1}{k}\sum_{i=1}^k\left(\frac{i}{k}\right)^{-{\hat\rho}}\right)
		\left(\frac{1}{k}\sum_{i=1}^k U_i\right)
		-
		\left(\frac{1}{k}\sum_{i=1}^k\left(\frac{i}{k}\right)^{-{\hat\rho}} U_i\right)
		\right]
	}{
		\left(\frac{1}{k}\sum_{i=1}^k\left(\frac{i}{k}\right)^{-{\hat\rho}}\right)
		\left(\frac{1}{k}\sum_{i=1}^k\left(\frac{i}{k}\right)^{-{\hat\rho}} U_i\right)
		-
		\left(\frac{1}{k}\sum_{i=1}^k\left(\frac{i}{k}\right)^{-2{\hat\rho}} U_i\right)
	}.
\end{equation}

Throughout the paper, the pair $(\hat{\rho},\hat{\beta})$ denotes the estimators defined in
\eqref{rho} and \eqref{beta}, evaluated at a high threshold level $k_1$, that is,
\((\hat{\rho},\hat{\beta})
\equiv
(\hat{\rho}(k_1),\hat{\beta}(k_1)),\)
and such that
\[
(\hat{\rho}-\rho)\log n = o_p(1), \qquad
\hat{\beta} = \beta(1+o_p(1)).
\]
These estimators are subsequently used to derive the asymptotic behaviour of all classes of RB estimators under study.

\subsection{Nonlinear RB power mean  and $L_p$ estimators}

In the spirit of \citet{caeiro2005direct}, new classes of RB power-mean and $L_p$ estimators were introduced in \citet{pm_red2024} and {\citet{penalva2020couple}}, respectively, 
\begin{multline}\label{eq:RBPM}
	{\widehat{\xi}^{\, \rm RBPM}(p,k)\equiv 	\widehat{\xi}^{\, \rm RBPM}_{(\hat\rho, \hat\beta)}(p,k)}
	:=
	\widehat{\xi}^{\, \rm PM}(p,k)
	\left(
	1 -
	\frac{1-(1-\widehat{\rho})^{p}}
	{p\,\widehat{\rho}\,(1-\widehat{\rho})^{p}}
	\,\widehat{\beta}
	\left(\frac{n}{k}\right)^{\widehat{\rho}}
	\right),\\
	\quad p>-\dfrac12\,,  p \neq 0,
\end{multline}
and 
\begin{equation}\label{eq:RBL}
	\widehat{\xi}^{\, \rm RBL}(p,k)\equiv 	\widehat{\xi}^{\, \rm RBL}_{(\hat\rho, \hat\beta)}(p,k)
	:=
	\widehat{\xi}^{\, \rm L}(p,k)
	\left(
	1 -
	\frac{\widehat{\beta}}
	{(1-\widehat{\rho})^{p}}
	\left(\frac{n}{k}\right)^{\widehat{\rho}}
	\right),
	\quad p>\dfrac{1}{2},
\end{equation}
where $(\hat{\beta},\hat{\rho})$ are consistent estimators of the second-order parameters $(\beta,\rho)$ and $ 	\widehat{\xi}^{\, \rm RBPM}(1,k)\equiv 	\widehat{\xi}^{\, \rm RBL}(1,k)=\widehat{\xi}^{\rm RBH}(k)$, in \eqref{eq:CH}.

{To account for the particular behaviour of the power mean class at $p=0$, we also introduce the following RB power mean estimator:

\begin{multline}\label{eq:RBPM0}
	{\widehat{\xi}^{\, \rm RBPM}(0,k)\equiv 	\widehat{\xi}^{\, \rm RBPM}_{(\hat\rho, \hat\beta)}(0,k)} \!
	:= 
	\widehat{\xi}^{\, \rm PM}(0,k)
	\left(
	1 +\frac{\log(1-\widehat \rho)}{\widehat{\rho}}\,\widehat{\beta}
	\left(\frac{n}{k}\right)^{\widehat{\rho}}
	\right),
	\quad 
\end{multline}

a continuation, for $p=0$, of the estimator in \eqref{eq:RBPM}, since $\lim_{p\rightarrow 0} (1-x^p)/p= -\log x$.
{Although the estimators in \eqref{eq:RBPM}, \eqref{eq:RBL} {and \eqref{eq:RBPM0}} eliminate the leading bias term of order $A(n/k)$, higher-order biases are usually still present and may influence the finite-sample behaviour of the estimators. Correcting higher-order bias terms is possible, but difficult because each additional correction requires the estimation of additional tail parameters. Motivated by the need to modify higher-order bias terms without estimating additional parameters, \citet{caeiro2005direct} and \citet{cabral2016reduced} suggested replacing {$1-b_{\text{H}}(\hat \rho)\hat\beta(n/k)^{\hat\rho}$ by $\exp\{-b_{\text{H}}(\hat \rho)\hat\beta(n/k)^{\hat\rho}\}$ or by $2-\exp\{b_{\text{H}}(\hat \rho)\hat\beta(n/k)^{\hat\rho}\}$.} Following this same idea, we next propose nonlinear extensions of the RB power-mean and $L_p$ classes.

Asymptotically equivalent exponential versions of \eqref{eq:RBPM} {and \eqref{eq:RBPM0} are} defined, by
\begin{multline}\label{eq:RBPMexp}
	\widehat{\xi}^{\, \rm RBPM*}(p,k) \equiv 	\widehat{\xi}^{\, \rm RBPM*}_{(\hat\rho, \hat\beta)}(p,k) \!
	:= 
	\widehat{\xi}^{\, \rm PM}(p,k)
	\exp\!\left\{
	-
	\frac{1-(1-\widehat{\rho})^{p}}
	{p\,\hat{\rho}\,(1-\hat{\rho})^{p}}
	\,\widehat{\beta}
	\left(\frac{n}{k}\right)^{\hat{\rho}}
	\right\}, \quad 
\end{multline}
for $p>-1/2$, where the case $p=0$ follows by continuous extension, since $\lim_{p\rightarrow 0} (1-x^p)/p= -\log x$.

Going one step further, we introduce the nonlinear {RB} power mean estimators
\begin{equation}\label{eq:NRBPM}
\begin{aligned}
\widehat{\xi}^{\,\mathrm{NRBPM}}(p,k)
&\equiv
\widehat{\xi}^{\,\mathrm{NRBPM}}_{(\hat\rho,\hat\beta)}(p,k)
\\
&:=
\widehat{\xi}^{\,\mathrm{PM}}(p,k)
\left(
2-
\exp\!\left\{
\frac{1-(1-\hat{\rho})^{p}}
{p\,\hat{\rho}\,(1-\hat{\rho})^{p}}
\,\hat{\beta}
\left(\frac{n}{k}\right)^{\hat{\rho}}
\right\}
\right),
\end{aligned}
\end{equation}
for $p>-1/2$, under the same continuous extension convention at $p=0$.

As in the power mean case, nonlinear corrections can be introduced to further refine bias
properties. An exponential version of \eqref{eq:RBL} is defined by
\begin{equation}\label{eq:RBLexp}
	\widehat{\xi}^{\rm RBL*}(p,k) \equiv 	\widehat{\xi}^{\, \rm RBL*}_{(\hat\rho, \hat\beta)}(p,k) \!
	:=
	\widehat{\xi}^{\, \rm L}(p,k)
	\exp\!\left\{
	-
	\frac{\hat{\beta}}
	{(1-\hat{\rho})^{p}}
	\left(\frac{n}{k}\right)^{\hat{\rho}}
	\right\}. 
\end{equation}

We further introduce the nonlinear {RB} $L_p$ estimator
\begin{equation}\label{eq:NRBL}
\begin{aligned}
\widehat{\xi}^{\,\mathrm{NRBL}}(p,k)
&\equiv
\widehat{\xi}^{\,\mathrm{NRBL}}_{(\hat\rho,\hat\beta)}(p,k)
\\
&:=
\widehat{\xi}^{\,\mathrm{L}}(p,k)
\left[
2-
\exp\!\left\{
\frac{\hat{\beta}}
{(1-\hat{\rho})^{p}}
\left(\frac{n}{k}\right)^{\hat{\rho}}
\right\}
\right],
\qquad p>\frac12.
\end{aligned}
\end{equation}

The estimators \eqref{eq:NRBPM} and \eqref{eq:NRBL} constitute the main methodological contributions of this paper. Their asymptotic behaviour, as well as their relationship with the corresponding {RB} estimators, are formally established in
Section~\ref{sec:asymp}.
\medskip 
\begin{remark}
For $p=1$, the EVI-estimators \(\widehat{\xi}^{\, \rm NRBPM}(1,k)\)  and \(\widehat{\xi}^{\, \rm NRBL}(1,k)\), in \eqref{eq:NRBPM} and \eqref{eq:NRBL}, respectively,  are also generalizations of the EVI-estimator,
\begin{equation}\label{eq:NCH}
\widehat{\xi}^{\, \rm NRBH}(k) \equiv 	\widehat{\xi}^{\, \rm NRBH}_{(\hat\rho, \hat\beta)}(k) \!
	:=
	\widehat{\xi}^{\, \rm RBH}(k)
	\left(
	2 -
	\exp\!\left\{
	\frac{\hat{\beta}}
	{1-\hat{\rho}}
	\left(\frac{n}{k}\right)^{\hat{\rho}}
	\right\}
	\right),
\end{equation}
introduced in \citet{cabral2016reduced}, with \(\widehat{\xi}^{\, \rm RBH}(k)\) the EVI-estimator defined  in \eqref{eq:CH}.
\end{remark}

\section{Asymptotic results}\label{sec:asymp}

This section presents the asymptotic behaviour of the RB estimators introduced in
Section~\ref{sec:rb}.  We briefly recall the asymptotic behaviour of the RB power mean, $L_p$ and Hill estimators, which will serve as a benchmark for the nonlinear corrections considered below. 
{Although the asymptotic behaviour of  the aforementioned estimators has been established in the literature (\citet{pm_red2024, penalva2016extreme, caeiro2005direct}), the results are presented here within a unified framework, extending the asymptotic distributional representation of $\widehat{\xi}^{\rm PM}(p,k)$ from $p>0$ to $p>-1/2$.}

The asymptotic results presented in Theorems \ref{theo:RB_unified} and \ref{theo:NRB_known} can be applied to any EVI estimator, $\tilde{\xi}(k)$, admitting a third-order asymptotic representation within the Hall class, of the type 
\begin{equation}\label{eq:EVI_dist}
\tilde{\xi}(k)\overset{d}{=}\xi+\frac{\tilde \sigma}{\sqrt{k}}\,\tilde Z_k + \tilde b_1\,A^2(n/k)+ \tilde b_2\,A^2(n/k)\bigl(1+o_p(1)\bigr),
\end{equation}
where $\tilde \sigma$, $\tilde b_1$, $\tilde b_2$ are estimating specific constants and $A(t)=\xi\beta\, t^\rho$, providing general expressions for the higher-order bias coefficients of both linear and nonlinear RB corrections.

\begin{theorem}[{\citet{pm_red2024, penalva2016extreme, caeiro2005direct}}]\label{theo:RB_unified}
	Assume the validity of the TOC, in \eqref{eq:TOC}, and of \eqref{eq:AB}, consider intermediate sequences $k=k_n$ satisfying \eqref{intermediate}, $1\le k<n$, and work in $\mathcal{S_{\bullet}}$, $\bullet=\{\rm{{H}, PM, L}\}$, where {$\mathcal{S_{\rm{H}}}=\{\xi: \xi>0\}$}, $\mathcal{S_{\rm{PM}}}=\{(\xi,p): \xi>0, p>-1/2\}$,  $\mathcal{S_{\rm{L}}}=\{(\xi,p): \xi>0, p>1/2\}$: 
	
	\medskip
	
	\noindent
	\textbf{(i) Known $(\rho,\beta)$.}  
	If the second-order parameters $(\rho,\beta)$ are known, the RB estimators in \eqref{eq:CH}, \eqref{eq:RBPM} and \eqref{eq:RBL}, with $(\hat\rho, \hat\beta)$ replaced by $(\rho, \beta)$,
	admit the following distributional representation:
	\begin{equation}\label{eq:RB_dist}
		\widehat\xi^{\, {\rm RB\bullet}}_{(\rho,\beta)}(p, k)
		\overset{d}{=}\,
		\xi+\frac{\sigma_{\bullet}}{\sqrt{k}}Z_k^{\bullet}
		+b_{\rm {RB}\bullet}A^2(n/k)(1+o_p(1)),
	\end{equation}
	with {$\sigma^2_{\bullet}$}  given in \eqref{eq:var_H}, \eqref{eq:var_PM} and \eqref{eq:var_L}, respectively,  and $Z_k^\bullet$ the standard normal random variables in \eqref{eq:HillAR}, \eqref{eq:PMAR} and \eqref{eq:LpAR}, $$b_{\rm RB\bullet}\equiv b_{\rm RB\bullet}(\rho):= b'_{\bullet}(\rho) - \frac{b^2_\bullet (\rho)}{\xi},$$ with $(b_\bullet(\rho), b'_\bullet(\rho))$ given in \eqref{eq:b_H} and \eqref{eq:b'_H} for the Hill estimator in \eqref{eq:Hill}, \eqref{eq:b_PM} and \eqref{eq:b'_PM} for the power mean-type estimator in \eqref{eq:RBPM}, and in \eqref{eq:b_L} and \eqref{eq:b'_L}, for the $L_p$ type estimator in \eqref{eq:RBL}.
Then:
	\begin{itemize}
		\item if $\sqrt{k}A(n/k)\to\lambda$ (finite or infinite) and $\sqrt{k}A^2(n/k)\to 0$,  
		\[
		\sqrt{k}\big(\widehat\xi^{\, {\rm RB\bullet}}_{(\rho,\beta)}(p, k)-\xi\big)
		\overset{d}{\longrightarrow}\mathcal{N}(0,\sigma^2_\bullet);
		\]
		\item if $\sqrt{k}A(n/k)\to\infty$ and $\sqrt{k}A^2(n/k)\to\lambda_A$,  
		\[
		\sqrt{k}\big(\widehat\xi^{\, {\rm RB\bullet}}_{(\rho,\beta)}(p, k)-\xi\big)
		\overset{d}{\longrightarrow}\mathcal{N}(\lambda_A b_{\rm RB\bullet},\sigma^2_\bullet).
		\]
	\end{itemize}
	
	\medskip
	
	\noindent
	\textbf{(ii) Estimated $(\rho,\beta)$.}  
	Assume now that $(\rho,\beta)$ are replaced by estimators $(\hat\rho,\hat\beta)\equiv(\hat\rho(k_1),\hat\beta(k_1))$, where $k_1$ is such that $\sqrt{k_1}A(n/k_1)\to\infty$, $k=o(k_1)$ and $(\hat\rho-\rho)\log n=o_p(1)$.  Then the estimator $\widehat\xi^{\,{\rm RB\bullet}}(p, k)$ is consistent for $\xi$.  
	If $\sqrt{k}A(n/k)\to\lambda$ finite (possibly zero), then
	\[
	\sqrt{k}\big(\widehat\xi^{\,{\rm RB\bullet}}(p, k)-\xi\big)
	\overset{d}{\longrightarrow}\mathcal{N}(0,\sigma^2_\bullet).
	\]
	This result remains valid when $\sqrt{k}A(n/k)\to\infty$ and $\sqrt{k}A^2(n/k)\to\lambda_A$, provided that
	\[
	(\hat\rho-\rho)\log (k/k_1)=o_p\!\left(\frac{1}{\sqrt{k}A(n/k)}\right).
	\]
\end{theorem}

\medskip 

\begin{remark} 
The nonlinear RB estimators in \eqref{eq:RBPMexp} and \eqref{eq:RBLexp} are asymptotically equivalent to the RB estimators in \eqref{eq:RBPM} and \eqref{eq:RBL}, respectively, under the validation of the SOC, in \eqref{eq:SOC}, and therefore the asymptotic behaviour  presented in {\bf Theorem \ref{theo:RB_unified}} holds for the RB estimators in \eqref{eq:RBPMexp} and \eqref{eq:RBLexp}, respectively. 
	\end{remark}

\medskip 

We now present the asymptotic distributional representation of  the nonlinear
RB estimators introduced in \eqref{eq:NRBPM} and \eqref{eq:NRBL}.

\begin{theorem}\label{theo:NRB_known}
Assume the conditions of {\bf Theorem~\ref{theo:RB_unified}} hold.  The nonlinear {RB}
power mean estimators $\widehat{\xi}^{\rm NRB\bullet}(p,k)$, in \eqref{eq:NRBPM} and \eqref{eq:NRBL}, with $\bullet=\{\rm PM, \rm L\,\}$, have the following asymptotic distributional  representation, assuming that $(\beta, \rho)$ are known:

	\begin{equation}\label{eq:NRB_dist}
	\widehat\xi^{\, {\rm NRB\bullet}}_{(\rho,\beta)}(p, k)
	\overset{d}{=}\,
	\xi+\frac{\sigma_{\bullet}}{\sqrt{k}}Z_k^{\bullet}
	+b_{\rm {NRB}\bullet}¢A^2(n/k)(1+o_p(1)),
\end{equation}
with $\sigma^2_{\bullet}$ given in \eqref{eq:var_PM} and \eqref{eq:var_L}, respectively,  and $Z_k^\bullet$  the  standard normal random variables in \eqref{eq:PMAR} and \eqref{eq:LpAR}, $$b_{\rm NRB\bullet}\equiv b_{\rm NRB\bullet}(\rho):= b'_{\bullet}(\rho) - \frac{3\, b^2_\bullet(\rho)}{2\xi},$$ with $(b_\bullet(\rho), b'_\bullet(\rho))$ given in \eqref{eq:b_PM} and \eqref{eq:b'_PM} for the power mean-type estimator in \eqref{eq:RBPM}, and in \eqref{eq:b_L} and \eqref{eq:b'_L}, for the $L_p$ type estimator in \eqref{eq:RBL}.
\end{theorem}

\begin{proof}
First, note that with $A(n/k) = \xi\,\beta(n/k)^{\rho}$, the estimators in \eqref{eq:NRBPM} and  \eqref{eq:NRBL} can be written as 
\begin{multline*}
 \widehat{\xi}^{\, \rm NRB\bullet}_{(\rho,\beta)}(p,k) \
:=
\widehat{\xi}^{\, \bullet}(p,k)
\left(
2 -
\exp\!\left\{
b_{\bullet}(\rho)
\,{\beta}
\left(\frac{n}{k}\right)^{\rho}
\right\}
\right) \\
= \widehat{\xi}^{\, \bullet}(p,k)
\left(
2 -
\exp\!\left\{
\frac{b_{\bullet}(\rho)}{\xi}A(n/k)
\right\}
\right).
\end{multline*}

Using the Taylor expansion, 
\(\exp(x) = 1 + x + \frac{x^2}{2} + o(x^2)\), when \(x \to 0\),
we obtain the following expansion for the bias term, 
\begin{align*}
2 - \exp\!\left(\frac{b_\bullet(\rho)}{\xi}A(n/k)\right)
&= 2 - \left(1 + \frac{b_\bullet(\rho)}{\xi}A(n/k) + \frac{b^2_\bullet(\rho)}{2\xi^2}A^2(n/k)\right)(1+o(1)) \\
&= 1 - \frac{b_\bullet(\rho)}{\xi}A(n/k) - \frac{b^2_\bullet(\rho)}{2\xi^2}A^2(n/k)(1+o(1)).
\end{align*}
		
Next, we combine this expansion with the third--order asymptotic distributions of the PM and $L_p$ estimators in \eqref{eq:PMAR} and \eqref{eq:LpAR}, respectively. With $\bullet=\{\rm PM, \rm L\}$, this gives
\begin{align*}
\hat{\xi}^{\,\mathrm{NRB\bullet}}_{(\rho,\beta)}(k)
\stackrel{d}{=}\;\xi
&+ \frac{\xi}{\sqrt{k}} Z^\bullet_k+ b_\bullet(\rho) A(n/k)+ b'_\bullet A^2(n/k)(1+o_p(1)) \\
&\quad \times \left[
1 - \frac{b_\bullet(\rho)}{\xi}A(n/k) - \frac{b^2_\bullet(\rho)}{2\xi^2}A^2(n/k)(1+o_p(1))\right].
\end{align*}
		
After simplification, we obtain
\[\hat{\xi}^{\,\mathrm{NRB\bullet}}_{(\rho,\beta)}(k)\stackrel{d}{=}\xi+ \frac{\xi}{\sqrt{k}} Z^\bullet_k +\left(b'_\bullet (\rho)-\frac{3\, b^2_\bullet(\rho)}{2\xi}\right)A^2(n/k)(1+o_p(1)).\]
\end{proof}

\medskip 
\begin{remark}
The asymptotic behaviour stated in {\bf Theorem \ref{theo:NRB_known}} remains valid for the estimator $\widehat{\xi}^{\rm NRBH}(k)$ defined in \eqref{eq:NCH}. 
\end{remark}
\medskip 

In the next Corollary we state the asymptotic equivalence between the classes of RB estimators RB$\bullet$ and NRB$\bullet$, {which} is also true for the classes RBH and NRBH.

\begin{corollary}\label{cor:NRB_RB}
Assume that the conditions of {\bf Theorem~\ref{theo:NRB_known}} hold. The nonlinear {RB} estimators $\widehat{\xi}^{\rm NRB\bullet}_{(\rho,\beta)}(p,k)$, $\bullet=\{\rm PM, \rm L\}$, defined in \eqref{eq:NRBPM} and \eqref{eq:NRBL} satisfy
\[\widehat{\xi}^{\rm NRB\bullet}_{(\rho,\beta)}(p,k) - \widehat{\xi}^{\rm RB\bullet}_{(\rho,\beta)}(p,k)\stackrel{d}{=}
- \frac{b^2_\bullet(\rho)}{2\xi}\,A^2(n/k)\left(1+o_p(1)\right).\]
Consequently, $\widehat{\xi}^{\rm NRB\bullet}_{(\rho,\beta)}(p,k)$ and $\widehat{\xi}^{\rm RB\bullet}_{(\rho,\beta)}(p,k)$ share the same asymptotic distribution, while differing in their second-order bias terms.
\end{corollary}

\medskip 
\begin{remark}
The results of {\bf Theorem \ref{theo:RB_unified}{\textit{(i)}}} and {\bf Theorem \ref{theo:NRB_known}} show that nonlinear RB estimators provide refinements at the level of second-order bias terms, without affecting asymptotic variance. These refinements may lead to improved finite-sample performance, particularly in scenarios where bias terms share the same sign, a phenomenon that is further investigated in the simulation study.
\end{remark}
\medskip

\begin{proposition}
Under the assumptions of {\bf Theorem \ref{theo:NRB_known}}, let $(\hat{\rho}, \hat{\beta})$ be consistent estimators of $(\rho,\beta)$ computed at an intermediate level $k_1$ such that $k = o(k_1)$. Assume further that $(\hat{\rho} - \rho)\log n = o_p(1)$ and that \((\hat\beta-\beta)/\beta \stackrel{p}{\sim} -(\hat{\rho}-\rho)\, \log\!\left(n/k_1\right),\) a condition satisfied by the $\beta$-estimator considered.	Then, the {\rm RB} estimator $\hat{\xi}^{\mathrm{NRB\bullet}}(p, k)\equiv \hat{\xi}^{\mathrm{NRB\bullet}}_{(\hat{\rho},\hat{\beta})}(p, k)$, with  $\bullet=\{\rm PM, \rm L\}$, satisfy
\[\hat{\xi}^{\mathrm{NRB\bullet}}_{(\hat{\rho},\hat{\beta})}(p, k)
\stackrel{d}{=}\;
\hat{\xi}^{\mathrm{NRB\bullet}}_{(\rho,\beta)}(p, k)-
b_\bullet(\rho)\, 
\bigl(\hat{\rho}-\rho\bigr)	\log\!\left(\frac{k_1}{k}\right)A\!\left(\frac{n}{k}\right)
\bigl(1+o_p(1)\bigr).\]
Under this assumption, we obtain
\begin{eqnarray*}
    \hat{\xi}^{\mathrm{NRB\bullet}}_{(\hat{\rho},\hat{\beta})}(p, k)-
\hat{\xi}^{\mathrm{NRB\bullet}}_{(\rho,\beta)}(p, k)
& \stackrel{p}{\sim} & 
\hat{\xi}^{\mathrm{RB\bullet}}_{(\hat{\rho},\hat{\beta})}(p, k)-
\hat{\xi}^{\mathrm{RB\bullet}}_{(\rho,\beta)}(p, k)\\
& \stackrel{p}{\sim}& -
b_\bullet(\rho)\, A(n/k)\,(\hat{\rho}-\rho)\,\log\!\left(\frac{k_1}{k}\right),
\end{eqnarray*}
with $\stackrel{p}{\sim}$ denoting asymptotic equivalence in probability. As a consequence, both estimators $\hat{\xi}^{\mathrm{RB\bullet}}_{(\hat{\rho},\hat{\beta})}(k)$
and $\hat{\xi}^{\mathrm{NRB\bullet}}_{(\hat{\rho},\hat{\beta})}(k)$
are consistent provided that $(\hat{\rho} -\rho)\, \log\left(k/k_1\right)=o_p\left(1/A(n/k)\right)$.
Moreover, these estimators admit an asymptotic distribution if the condition
\((\hat{\rho}-\rho)\,\log\!\left(k/k_1\right)=o_p\!\left(1/(\sqrt{k}\,A(n/k))\right)\)holds.
\end{proposition}

\begin{proof}
By definition, the NRB$\bullet$ estimators can be written as
\[\hat{\xi}^{\mathrm{NRB_\bullet}}_{(\hat{\rho},\hat{\beta})}(p, k)=
\hat{\xi}^{\bullet}_{(\hat{\rho},\hat{\beta})}(p, k)
\left(2 -\exp\!\left(b_\bullet(\hat \rho)\hat{\beta}\left(\frac{n}{k}\right)^{\hat{\rho}}\right)
\right).\]
Setting $$c_\bullet(\rho, \beta):=\exp\!\left(b_\bullet( \rho){\beta}\left(\frac{n}{k}\right)^{{\rho}}\right),$$
the {delta method} enables us to write 
\begin{multline*}c_\bullet(\hat \rho, \hat \beta)=c_\bullet(\rho,\beta)+(\hat \beta-\beta)\frac{\partial c_\bullet(\rho,\beta)}{\partial \beta}(1+o_p(1)) \\
+(\hat \rho-\rho)\frac{\partial c_\bullet(\rho,\beta)}{\partial \rho}(1+o_p(1)),
\end{multline*}
where 
\[\frac{\partial c_\bullet(\rho,\beta)}{\partial \beta}=
b_\bullet(\rho)\left(\frac{n}{k}\right)^{\rho}\, c_\bullet(\rho, \beta)\]
and 
\begin{eqnarray*}
		\frac{\partial c_\bullet(\rho,\beta)}{\partial \rho}&=&
		\left(\frac{\partial\, b_\bullet(\rho)}{\partial\,  \rho}\beta\left(\frac{n}{k}\right)^{\rho}+
	b_\bullet(\rho)\beta\left(\frac{n}{k}\right)^{\rho}\log \left(\frac{n}{k}\right)\right)c_\bullet(\rho, \beta)\\
		&=& c_\bullet(\rho, \beta)\beta\left(\frac{n}{k}\right)^{\rho}\left(\frac{\partial\, b_\bullet(\rho)}{\partial\,  \rho}+b_\bullet(\rho)\log \left(\frac{n}{k}\right)\right)\\
		&=&c_\bullet(\rho, \beta)\, b_\bullet(\rho)\, \beta\left(\frac{n}{k}\right)^{\rho}\left(f_\bullet(\rho)+\log \left(\frac{n}{k}\right)\right),
	\end{eqnarray*}
	with $f_\bullet(\rho):=\frac{1}{b_\bullet(\rho)}\frac{\partial\, b_\bullet(\rho)}{\partial\,  \rho}.$ Since \(k=o(k_1)\), \(f_\bullet(\rho)=o_p\left(\log\left(\frac{k_1}{k}\right)
\right)\) and 
\begin{eqnarray*}
\frac{\partial c_\bullet(\rho,\beta)}{\partial \rho}&\sim &c_\bullet(\rho, \beta)\, b_\bullet(\rho)\, \beta\left(\frac{n}{k}\right)^{\rho}\log \left(\frac{n}{k}\right).
\end{eqnarray*}
	Then, 
	\begin{eqnarray*}
	\hat{\xi}^{\mathrm{NRB_\bullet}}_{(\hat{\rho},\hat{\beta})}(p, k)
	&=&
	\hat{\xi}^{\bullet}_{(\hat{\rho},\hat{\beta})}(p, k)
	\left(
	2 -
	\exp\!\left(b_\bullet(\hat \rho)\hat{\beta}\left(\frac{n}{k}\right)^{\hat{\rho}}\right)
	\right)\\
    & = & \hat{\xi}^{\bullet}_{(\hat{\rho},\hat{\beta})}(p, k)
	\left(
	2 -c_\bullet(\hat\rho, \hat\beta)\right)\\
	&	=\;&	\hat{\xi}^{\bullet}_{(\hat{\rho},\hat{\beta})}(p, k)
	\left(
	2 -c_\bullet(\rho,\beta)-(\hat \beta-\beta)\frac{\partial c_\bullet(\rho,\beta)}{\partial \beta}(1+o_p(1))\right.
   \\
      & & \hspace{10pc} {\left. -(\hat \rho-\rho)\frac{\partial c_\bullet(\rho,\beta)}{\partial \rho}(1+o_p(1))\right).}
		\end{eqnarray*}
Using the asymptotic representation of the \(\bullet\)-type estimator, together with the asymptotic expansion of \(c_\bullet(\rho,\beta)\), and regrouping terms leads to
	
	\begin{multline*}
	\hat{\xi}^{\mathrm{NRB\bullet}}_{(\hat{\rho},\hat{\beta})}(p, k)
	\stackrel{d}{=}\;
	\hat{\xi}^{\mathrm{NRB\bullet}}_{(\rho,\beta)}(p, k)
	-b_\bullet(\rho)
	\left[\frac{\hat \beta-\beta}{\beta}+
	\bigl(\hat{\rho}-\rho\bigr)\log\!\left(\frac{n}{k}\right)
	\right]A\!\left(\frac{n}{k}\right)
	\bigl(1+o_p(1)\bigr).
\end{multline*}
Since 	$(\hat{\beta}-\beta)/\beta \stackrel{p}{\sim} -(\hat{\rho}-\rho)\, \log\!\left(n/k_1\right)$ (\citet{gomes2008tail}), we can write 

	\begin{align*}
		\hat{\xi}^{\mathrm{NRB\bullet}}_{(\hat{\rho},\hat{\beta})}(p, k)
		\stackrel{d}{=}\;&
		\hat{\xi}^{\mathrm{NRB\bullet}}_{(\rho,\beta)}(p, k)
		-
		b_\bullet(\rho)\,
	\bigl(\hat{\rho}-\rho\bigr)	\log\!\left(\frac{k_1}{k}\right)A\!\left(\frac{n}{k}\right)
		\bigl(1+o_p(1)\bigr),
	\end{align*}
	which completes the proof.
\end{proof}

\subsection{Asymptotic comparison at optimal levels}

Let $\tilde{\xi}(k)$ denote a generic {RB EVI-estimator that admits} the following distributional expansion:
\begin{equation}\label{eq:RB_dist}
\tilde{\xi}(k)\overset{d}{=}\xi+\frac{\tilde \sigma}{\sqrt{k}}\,\tilde Z_k + \tilde b_2\,A^2(n/k)\bigl(1+o_p(1)\bigr),
\end{equation}
where $\tilde \sigma$ and $\tilde b_2$ are estimating specific constants and $A(\cdot)$ denotes the second-order auxiliary function. Let ${\rm Bias}_{\infty}\!\left(\tilde{\xi}(k)\right)$ and
${\rm Var}_{\infty}\!\left(\tilde{\xi}(k)\right)$ denote,
respectively, the asymptotic bias and variance of the generic RB estimator
$\tilde{\xi}(k)$. The
asymptotic mean squared error (AMSE) of $\tilde{\xi}(k)$ can be expressed as
\begin{align}
	\label{eq:AMSE}
	{\rm AMSE}\!\left(\tilde{\xi}(k)\right)
	&=
	{\rm Var}_{\infty}\!\left(\tilde{\xi}(k)\right)
	+
	\left(
	{\rm Bias}_{\infty}\!\left(\tilde{\xi}(k)\right)
	\right)^2 \nonumber\\
	&=
	\frac{\tilde \sigma^2}{k}
	+
	\tilde b^2\,A^4(n/k),
	\qquad \tilde b\neq 0.
\end{align}

Let $\tilde{k}_0$ be the value of $k$ that minimizes ${\rm AMSE}\!\left(\tilde{\xi}(k)\right)$, in \eqref{eq:AMSE}. The {RB} estimator evaluated at this optimal level is denoted by $\tilde{\xi}_0 = \tilde{\xi}(\tilde{k}_0)$, and the associated limiting mean squared error is referred to as ${\rm LMSE}$.

Under standard regular variation assumptions ({\citet{bingham1987regular,bingham2007regular}}), and provided that $\tilde b \neq 0$, it follows that there exists a normalizing function $\varphi(n,\xi,\rho)$, depending only on the underlying model and not on the estimator itself ({\citet{deHaanPeng1998}}), such that
\[
\lim_{n\to\infty}
\varphi(n,\xi,\rho)\,
{\rm AMSE}\!\left(\tilde{\xi}_0\right)
=
\left(\tilde \sigma^2\right)^{-\frac{2\rho}{1-4\rho}}
\left(\tilde b^2\right)^{\frac{1}{1-4\rho}}
=: {\rm LMSE}\!\left(\tilde{\xi}_0\right).
\]

By considering ratios of LMSE values corresponding to different RB estimators, the normalizing function $\varphi(\cdot)$ cancels out, leading to an efficiency measure that depends solely on the leading bias and variance terms. Specifically, for two RB
estimators $\tilde{\xi}^{(1)}$ and $\tilde{\xi}^{(2)}$, we define the asymptotic relative efficiency (AREFF) as
\begin{equation}\label{eq:arref_def}
	{\rm AREFF}_{(1)|(2)}
	:=
	\sqrt{
		\frac{{\rm LMSE}(\tilde{\xi}_0^{(2)})}
		{{\rm LMSE}(\tilde{\xi}_0^{(1)})}
	}
	=
	\left(\frac{\tilde \sigma^{(2)}}{\tilde \sigma^{(1)}}\right)^{-\frac{4\rho}{1-4\rho}}
	\left|\frac{\tilde b^{(2)}}{\tilde b^{(1)}}\right|^{\frac{1}{1-4\rho}}.
\end{equation}

Values of ${\rm AREFF}_{(1)|(2)}$ greater than one indicate that
$\tilde{\xi}^{(2)}(k)$ is asymptotically more efficient than $\tilde{\xi}^{(1)}(k)$,
whereas values below one favour $\tilde{\xi}^{(1)}(k)$. Equality corresponds to equivalent asymptotic performance.

We consider two types of comparisons. First, each newly proposed RB estimator is compared with its corresponding original RB estimator. Second, all new estimators are compared with the benchmark RBH estimator defined in \eqref{eq:CH}.

\medskip
\begin{remark}
The threshold minimizing the AMSE of the RBH estimator, in \eqref{eq:CH}, is obtained as
\begin{equation}
	\label{eq:k0_CH}
	\tilde k_0^{\rm RBH}
	=
	\left(
	-4\rho
	\left(
	\frac{\delta}{1-2\rho}
	-
	\frac{1}{(1-\rho)^2}
	\right)^2
	\beta^4
	\right)^{-\frac{1}{1-4\rho}}
	n^{-\frac{4\rho}{1-4\rho}}.
\end{equation}
\end{remark}
\medskip 

In what follows, we take
\begin{eqnarray*}
\left(\hat{\xi}^{(1,p)}, \hat{\xi}^{(2,p)}\right)&\equiv &\left\{ 
\left({{\rm RBPM}(p)}, {\rm RBH}\right), 
\left({{\rm NRBPM}(p)}, {\rm NRBH}\right),  \right. \\
&&\left({{\rm NRBPM}(p)}, {\rm RBH}\right), \left({{\rm NRBPM}(p)}, {\,{\rm RBPM}(p)}\right),   \\
&& \left({{\rm NRBL}(p)}, {\rm NRBH}\right), \left({{\rm NRBL}(p)}, {\rm RBH}\right),\\
&& \left.\left({{\rm NRBL}(p)}, {\,{\rm RBL}(p)}\right)\right\}.
\end{eqnarray*}

The AREFF of the RBPM$(p)$ estimators relative to the RBH estimator was presented in \citet{pm_red2024} for $p>0$ and is extended here to $p>-1/2$. The resulting expressions for the AREFF are given by

\begin{multline}\label{eq:AREFF_NRBPMNCH}
{\rm AREFF}_{{\rm NRBPM}(p)\,|\,{\rm NRBH}}
=
	\left(
	\frac{1}{p}
	\sqrt{
		\frac{\Gamma(2p+1)}{\Gamma^2(p+1)}-1
	}
	\right)^{\frac{4\rho}{1-4\rho}}
	\times \\[2mm]
\left|
	\frac{
		2(p\rho)^2
		\left(
		\frac{\delta}{1-2\rho}
		-
{\frac{3}{2(1-\rho)^2}}
		\right)
	}{
		\frac{p\bigl(1-(2-\delta)\rho\bigr)}{(1-2\rho)^p}
		+
		\frac{2(2+p\rho)}{(1-\rho)^p}
		-
		\frac{p+2}{(1-\rho)^{2p}}
		-\delta p\rho -2
	}
	\right|^{\frac{1}{1-4\rho}}.
\end{multline}

\vspace{-0.2cm}

\begin{multline}\label{eq:AREFF_NRBPMCH}
	{\rm AREFF}_{{\rm NRBPM}(p)\,|\,{\rm RBH}}
	=
	\left(
	\frac{1}{p}
	\sqrt{
		\frac{\Gamma(2p+1)}{\Gamma^2(p+1)}-1
	}
	\right)^{\frac{4\rho}{1-4\rho}}
	\times \\[2mm]
	\left|
	\frac{
		2(p\rho)^2
		\left(
		\frac{\delta}{1-2\rho}
		-
		\frac{1}{(1-\rho)^2}
		\right)
	}{
		\frac{p\bigl(1-(2-\delta)\rho\bigr)}{(1-2\rho)^p}
		+
		\frac{2(2+p\rho)}{(1-\rho)^p}
		-
		\frac{p+2}{(1-\rho)^{2p}}
		-\delta p\rho -2
	}
	\right|^{\frac{1}{1-4\rho}}.
\end{multline}

\vspace{-0.2cm}
\begin{multline}\label{eq:AREFF_NRBPMRBPM}
	{\rm AREFF}_{{\rm NRBPM}(p)\,|\,{\rm RBPM(p)}}
	 =	\\
	 \left|
	\frac{
	\frac{p\bigl(1-(2-\delta)\rho\bigr)}{(1-2\rho)^p}
	+
	\frac{2(1+p\rho)}{(1-\rho)^p}
	-
	\frac{p+1}{(1-\rho)^{2p}}
	-
	\bigl(1+\delta p\rho\bigr)
	}{
		\frac{p\bigl(1-(2-\delta)\rho\bigr)}{(1-2\rho)^p}
		+
		\frac{2(2+p\rho)}{(1-\rho)^p}
		-
		\frac{p+2}{(1-\rho)^{2p}}
		-\delta p\rho -2
	}
	\right|^{\frac{1}{1-4\rho}}. 
\end{multline}

\vspace{-0.2cm}
\begin{multline}\label{eq:AREFF_NRBLCH}
{\rm AREFF}_{{\rm NRBL}(p)\,|\,{\rm NRBH}}
	=
	\left(
	\sqrt{
		\frac{\Gamma(2p-1)}{\Gamma^2(p)}
	}
	\right)^{\frac{4\rho}{1-4\rho}}
	\times \\
	\left|
	\frac{
	\rho
		\left(
		\frac{\delta}{1-2\rho}
		-
		\frac{3}{2(1-\rho)^2}
		\right)
	}{
\frac{1-2\rho+\delta\rho}{(1-2\rho)^{p}}-\frac{2+\rho-2\rho(1-\rho)^{p}}{2(1-\rho)^{2p}} }
	\right|^{\frac{1}{1-4\rho}}.
\end{multline}
\vspace{-0.5cm}

\begin{multline}\label{eq:AREFF_NRBLCH}
	{\rm AREFF}_{{\rm NRBL}(p)\,|\,{\rm RBH}}
	=
	\left(
	\sqrt{
		\frac{\Gamma(2p-1)}{\Gamma^2(p)}
	}
	\right)^{\frac{4\rho}{1-4\rho}}
	\times \\
	\left|
	\frac{
	\rho
		\left(
		\frac{\delta}{1-2\rho}
		-
		\frac{1}{(1-\rho)^2}
		\right)
	}{
\frac{1-2\rho+\delta\rho}{(1-2\rho)^{p}}-\frac{2+\rho-2\rho(1-\rho)^{p}}{2(1-\rho)^{2p}} }
	\right|^{\frac{1}{1-4\rho}}.
\end{multline}

\begin{equation}\label{eq:AREFF_NRBLRBL}
	{\rm AREFF}_{{\rm NRBL}(p)\,|\,{\rm RBL}(p)}
	= \left|
	\frac{
	\frac{1-2\rho+\delta\rho}{(1-2\rho)^{p}}-\frac{1-\rho(1-\rho)^{p}}{(1-\rho)^{2p}} 
	}{
		\frac{1-2\rho+\delta\rho}{(1-2\rho)^{p}}-\frac{2+\rho-2\rho(1-\rho)^{p}}{2(1-\rho)^{2p}}
	}
	\right|^{\frac{1}{1-4\rho}}.
\end{equation}

\medskip

For clarity, the analysis focuses on a reduced class of Pareto-type models. More specifically, in the Burr and Generalized Pareto cases (i.e., when $\delta = 1$), {seven} contour representations of
${\rm AREFF}_{{\hat{\xi}^{(1,p)}| \hat{\xi}^{(2,p)}}}$ are provided in Figures~\ref{fig:RBPM_CH}--\ref{fig:NRBL_RBL}.
The blue contour curves identify the pairs $(p,\rho)$ for which the two estimators attain the same asymptotic efficiency. The shaded areas correspond to regions of the parameter space where $\hat{\xi}^{(1)}(k)$ {is more efficient than} $\hat{\xi}^{(2)}(k)$. The red curve represents, for each fixed value of \(\rho\), the value of \(p\)
that maximizes the AREFF within the favourable region, {providing the theoretically optimal tuning of \(p\) as a function of \(\rho\),}
\[
p_{\mathrm{opt}}(\rho)
=
\arg\max_{p}
\operatorname{AREFF}_{\hat{\xi}^{(1,p)}| \hat{\xi}^{(2,p)}} \,\
\text{subject to }
\operatorname{AREFF}_{\hat{\xi}^{(1,p)}| \hat{\xi}^{(2,p)}}(\rho)>1.
\]

\begin{figure}[htbp]
	\centering
	\includegraphics[width=0.8\textwidth]{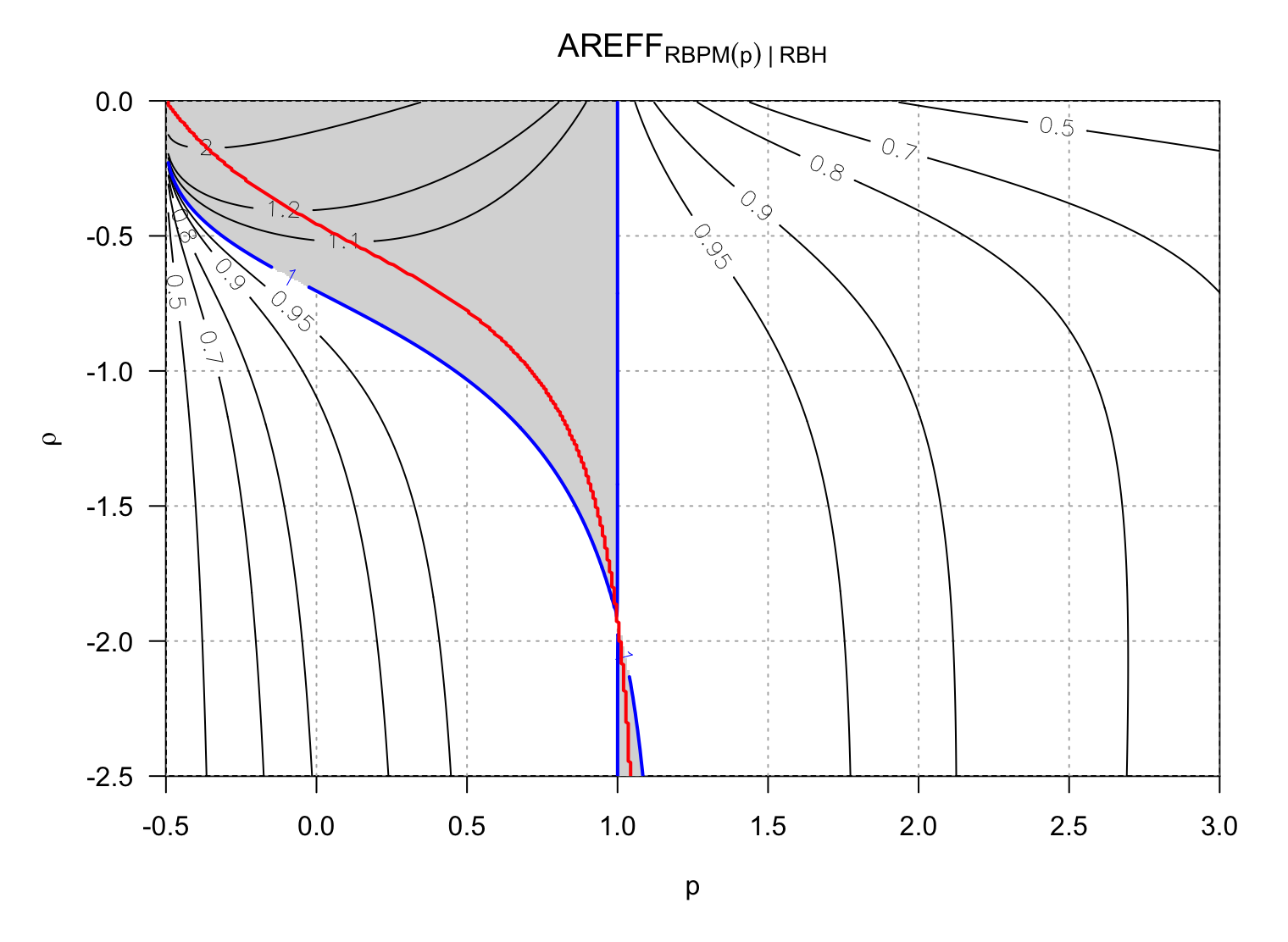}
	\caption{Contour plot of ${\rm AREFF}_{{\rm RBPM}(p)|{\rm RBH}}$, with $\delta=1$ in the $(p,\rho)$-plane.}
	\label{fig:RBPM_CH}
\end{figure}
\begin{figure}[htbp]
	\centering
	\includegraphics[width=0.8\textwidth]{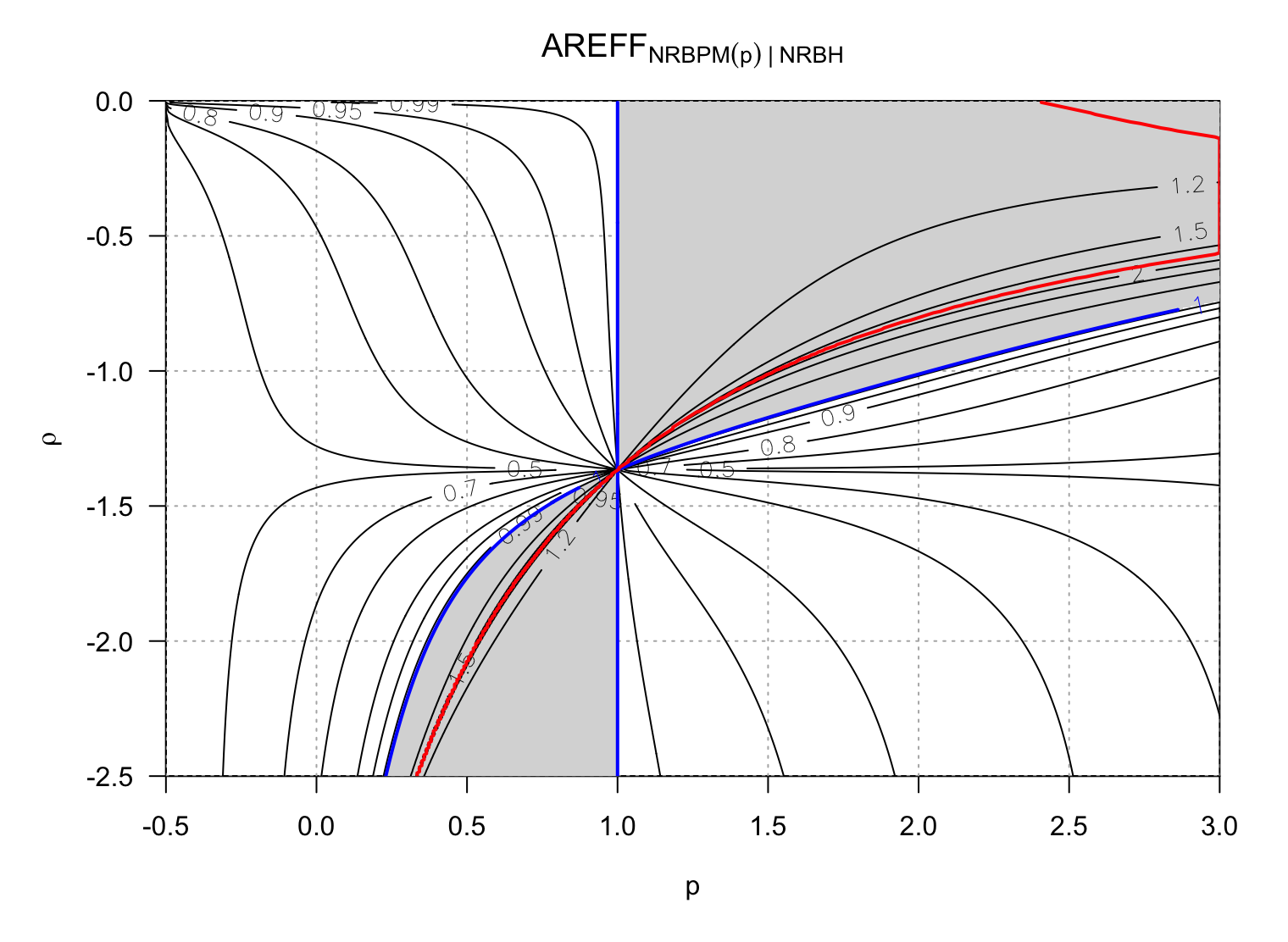}
	\caption{Contour plot of ${\rm AREFF}_{{{\rm NRBPM}}(p)|{\rm NRBH}}$, with $\delta=1$ in the $(p,\rho)$-plane.}
	\label{fig:NRBPM_NCH}
\end{figure}

\begin{figure}[htbp]
	\centering
	\includegraphics[width=0.8\textwidth]{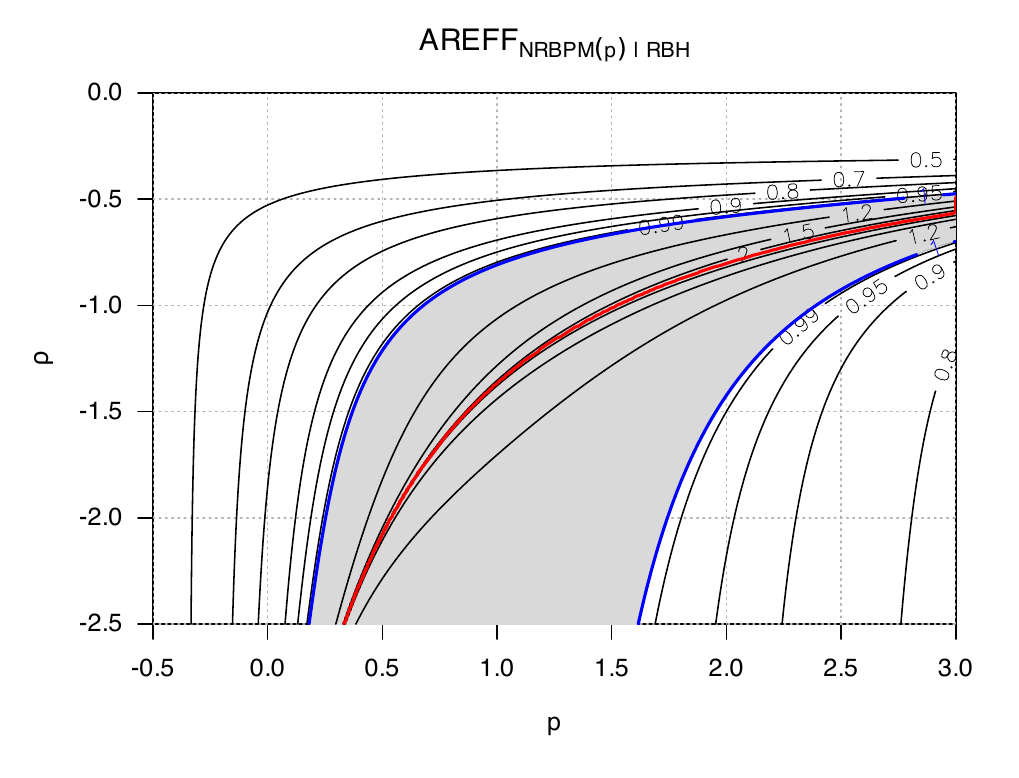}
	\caption{Contour plot of ${\rm AREFF}_{{\rm NRBPM}(p)|{\rm RBH}}$, with $\delta=1$ in the $(p,\rho)$-plane.}
	\label{fig:NRBPM_CH}
\end{figure}

\begin{figure}[htbp]
	\centering
	\includegraphics[width=0.8\textwidth]{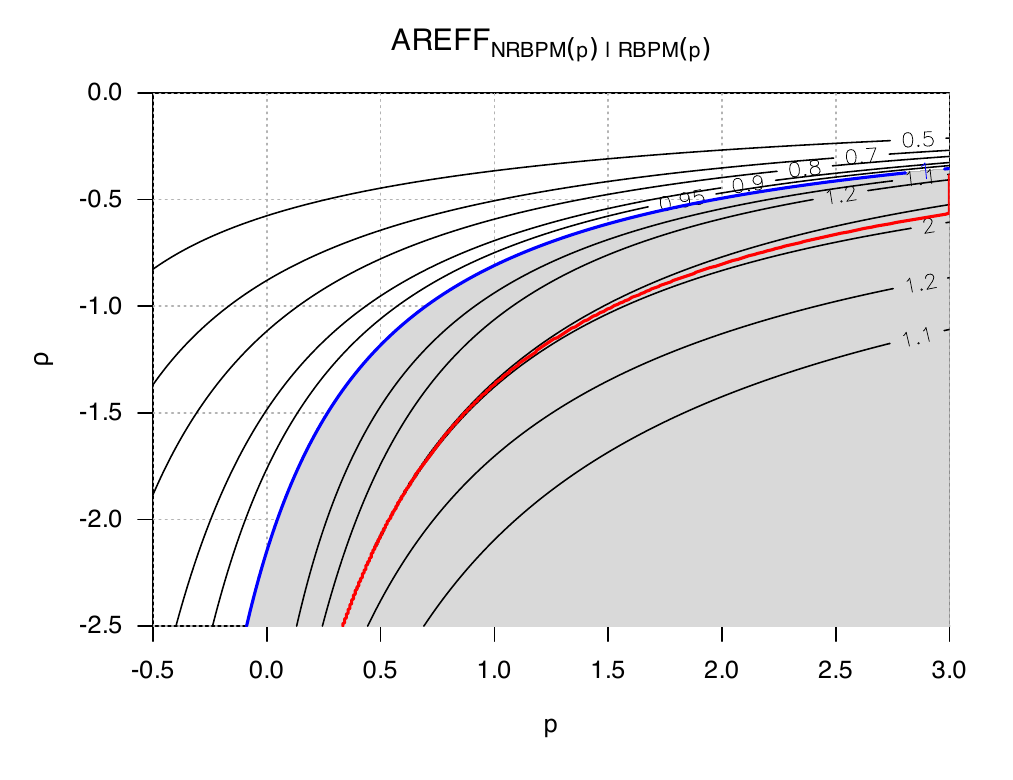}
	\caption{Contour plot of ${\rm AREFF}_{{\rm NRBPM}(p)|{\rm RBPM}(p)}$, with $\delta=1$ in the $(p,\rho)$-plane.}
	\label{fig:NRBPM_RBPM}
\end{figure}

\begin{figure}[htbp]
	\centering
	\includegraphics[width=0.8\textwidth]{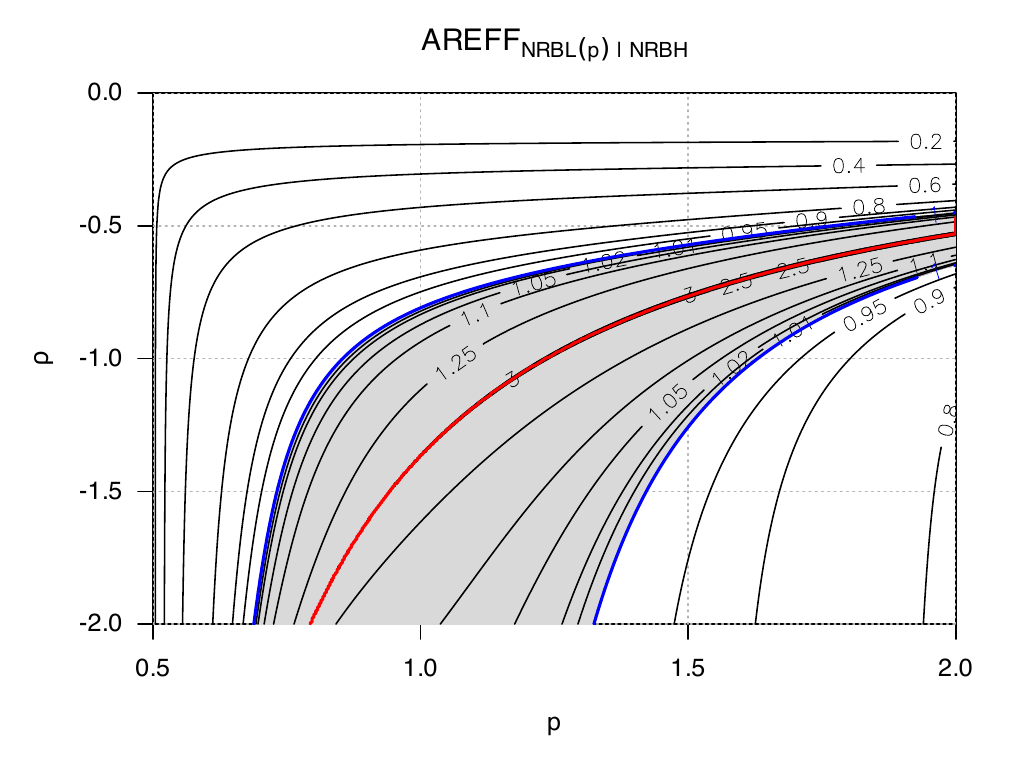}
	\caption{Contour plot of ${\rm AREFF}_{{\rm NRBL}(p)|{\rm NRBH }}$, with $\delta=1$ in the $(p,\rho)$-plane.}
	\label{fig:NRBL_NCH}
\end{figure}

\begin{figure}[htbp]
	\centering
	\includegraphics[width=0.8\textwidth]{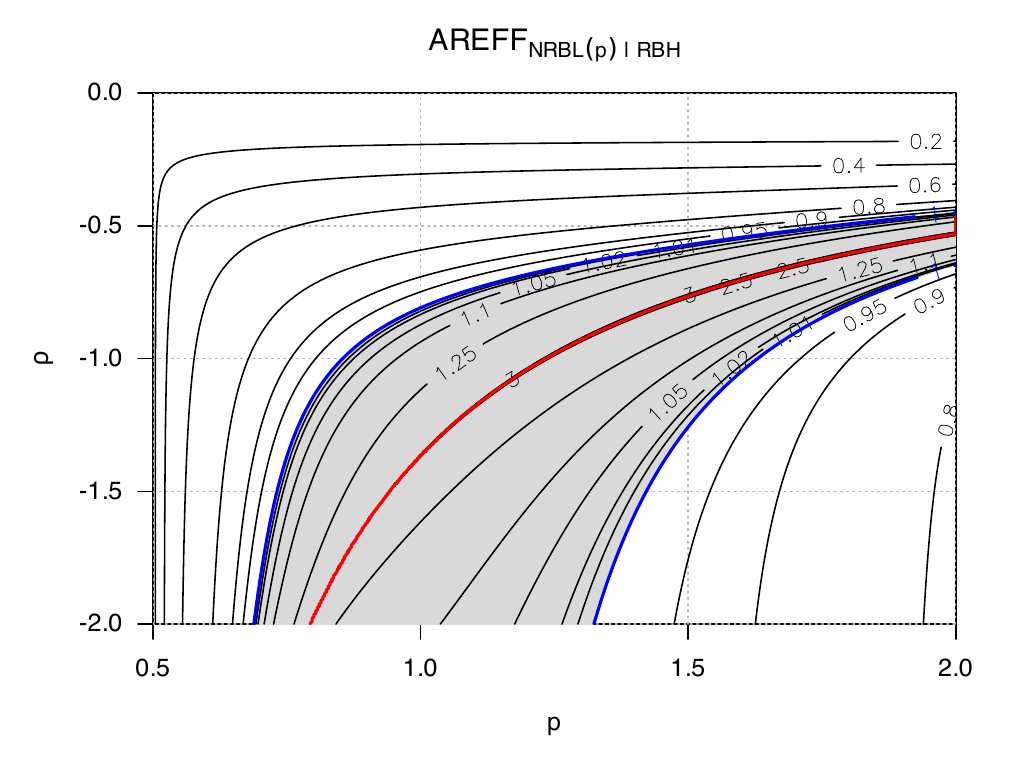}
	\caption{Contour plot of ${\rm AREFF}_{{\rm NRBL}(p)|{\rm RBH }}$, with $\delta=1$ in the $(p,\rho)$-plane.}
	\label{fig:NRBL_CH}
\end{figure}

\begin{figure}[htbp]
	\centering
	\includegraphics[width=0.8\textwidth]{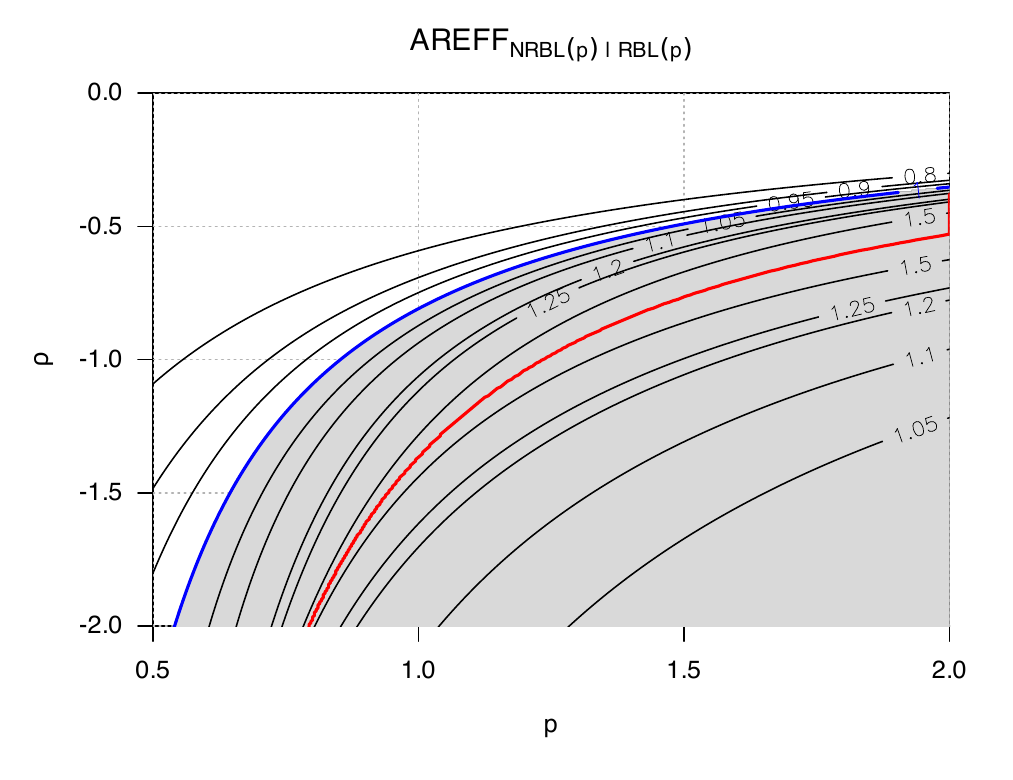}
	\caption{Contour plot of ${\rm AREFF}_{{\rm NRBL}(p)|{{\rm RBL}(p)}}$, with $\delta=1$ in the $(p,\rho)$-plane.}
	\label{fig:NRBL_RBL}
\end{figure}

The contour plots reveal that the proposed nonlinear RB estimators outperform their classical and linear counterparts over substantial regions of the $(p,\rho)$-plane. The gains are particularly pronounced for moderately negative values of the second-order parameter $\rho$, whereas they tend to diminish as $\rho$ approaches zero. Moreover, the optimal value of $p$ varies systematically with $\rho$, confirming the importance of an appropriate choice of the tuning parameter. In particular, the comparisons of NRBPM versus RBPM, in Figure \ref{fig:NRBPM_RBPM}, and NRBL versus RBL, in Figure \ref{fig:NRBL_RBL}, show that the nonlinear transformation can provide an efficiency gain beyond the corresponding linear RB estimator.

\newpage

\section{Monte Carlo simulation study}\label{sec:mc}

To evaluate the finite-sample performance of the classes of estimators, \\
$\widehat{\xi}^{\, \rm RBPM}(p,k)$, $\widehat{\xi}^{\, \rm NRBPM}(p,k)$, $\widehat{\xi}^{\, \rm RBL}(p,k)$ and $\widehat{\xi}^{\, \rm NRBL}(p,k)$, we conduct a Monte Carlo simulation study across selected tuning parameters $p$. For baseline comparisons, we include both the classic Hill estimator in \eqref{eq:Hill} and the RB Hill ($p=1$), denoted by ${\rm RBH}\equiv{\rm RBPM}(1)\equiv{\rm RBL}(1)$. The values of $p$ were selected based on findings from a preliminary study. The following performance metrics were computed using $r = 2500$ Monte Carlo replicates for sample sizes $n \in \{50, 100, 250, 500, 1000, 2500, 5000\}$. 
For each estimator, generally denoted by $\widehat{\xi}^{\bullet}(p,k)$, and sample size we obtained the simulated mean value (E) and root mean squared error (RMSE):
\begin{equation}
\begin{split}
\mathrm{E}(\widehat{\xi}^{\bullet}(p,k)) &= \frac{1}{r}\sum_{j=1}^{r}\widehat{\xi}_j^{\bullet}(p,k),\\
\mathrm{RMSE}(\widehat{\xi}^{\bullet}(p,k)) &= 
\sqrt{\frac{1}{r}\sum_{j=1}^{r}
\left(\widehat{\xi}_j^{\bullet}(p,k)-\xi\right)^2},
\label{eq_mv_rmse}
\end{split}
\end{equation}
with 
 $1\leq k \leq n-1$ and $j$ denoting the $j$-th sample. We have further computed the simulated optimum level 
\begin{equation}
\hat{k}_0^{\bullet}= \arg \min_k {\rm RMSE}[\hat{\xi}^{\bullet}(p,k)],
\label{eq:k0}%
\end{equation}
and the simulated characteristics in \eqref{eq_mv_rmse} at the optimal level in \eqref{eq:k0}.

Data were generated from the following heavy-tailed models:
\begin{itemize}
    \item Fréchet Model: $F(x) = \exp(-x^{-1/\xi})$, $x > 0$, with parameter $\xi = 0.5$ ($\rho = -1$, $\delta = 5/3$).
    \item Burr Model: $F(x) = 1 - (1 + x^{-\rho/\xi})^{1/\rho}$, $x > 0$, with parameters $\xi=1$, and $\rho= -0.8$ ($\delta = 1$).
    \item GP model: $F(x) = 1 - (1 + \xi x)^{-1/\xi}$, $x > 0$, with parameter $\xi=1$ ($\rho = -\xi$, $\delta = 1$).
    \item Half-$t_\nu$ Model: The absolute value of a Student's-$t$ distribution with $\nu =2$ degrees of freedom, corresponding to $(\xi, \rho) = (1/\nu, -2/\nu)$ with parameter $\delta = \nu(\nu^2 + 4\nu + 2) / [(\nu + 4)(\nu + 1)]$.
\end{itemize}

\subsection{Simulated results for all thresholds}%
We first evaluated the estimator stability as a function of all possible thresholds, $k$, using simulated samples of size $n = 1000$. Figures \ref{fig_frechet}--\ref{fig_halft2} show the sample paths of $E$ (left) and $\text{RMSE}$ (right). The horizontal line in the left plot marks the true EVI value. For all models, the Hill estimators mean value
sample path departs from the true value almost immediately and rises linearly with $k$, while the RMSE exhibits a sharp peak. 
For all models, the RBH and all RB/NRB variants have a stability plateau close to the true EVI, before curving upward again until $k$ approaches $1000$. This is mirrored in the RMSE plot, whereas the RB/NRB estimators trace a much flatter, wider mini\-mum across a wide range of  $k$-values. For the Fréchet model,  the mean value and RMSE pattern of all RB estimators nearly overlap, suggesting that for this particular model the choice of tuning parameter $p$ has little impact. In Figures \ref{fig_burr1}–\ref{fig_halft2}, some visible separation emerges among the RB/NRB sample paths,  most notably for the Burr and Half-$t_2$ models. For those models, the NRBPM/NRBL curves extended the stability region in the mean-value plot, remaining close to the true  EVI over a substantially wider range of $k$ values and the RMSE tends to reach the smallest value.

\begin{figure}[h!]
\includegraphics[width=.475\textwidth]{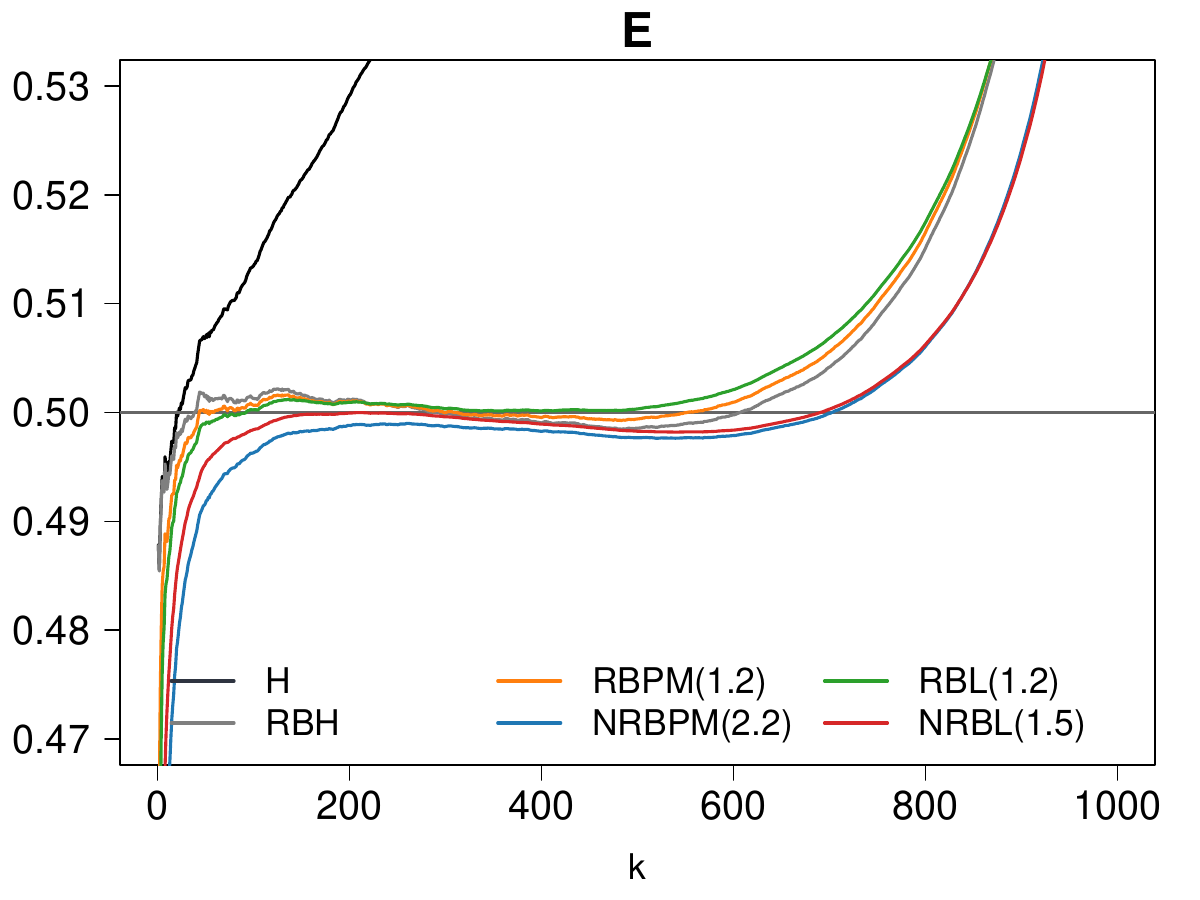}\:\:
\includegraphics[width=.475\textwidth]{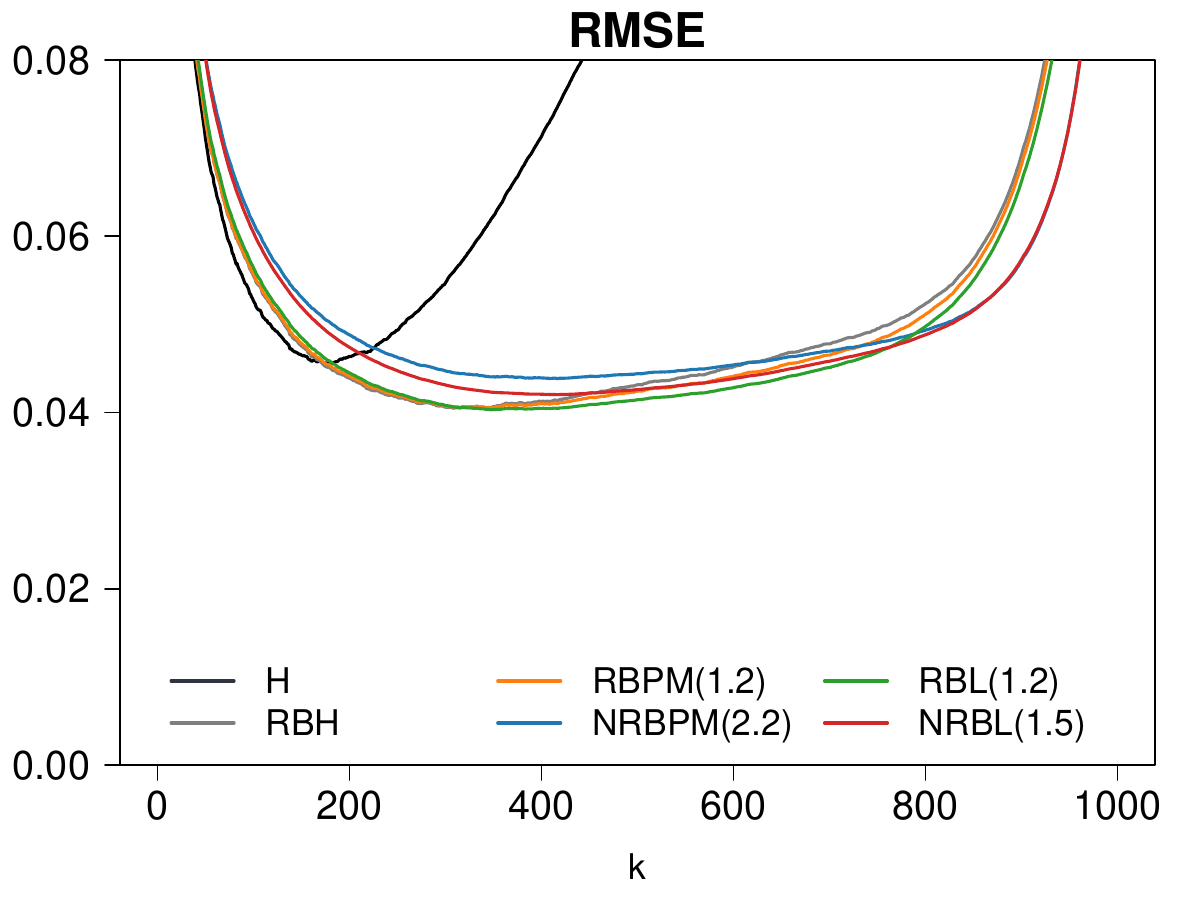}\\
\caption{%
Simulated mean values (left) and RMSE (right) of the EVI-estimators
under study for samples of size 1000 from the Fr\'echet model, with $\xi=0.5$
}
\label{fig_frechet}
\end{figure}

\begin{figure}[h!]
\includegraphics[width=.475\textwidth]{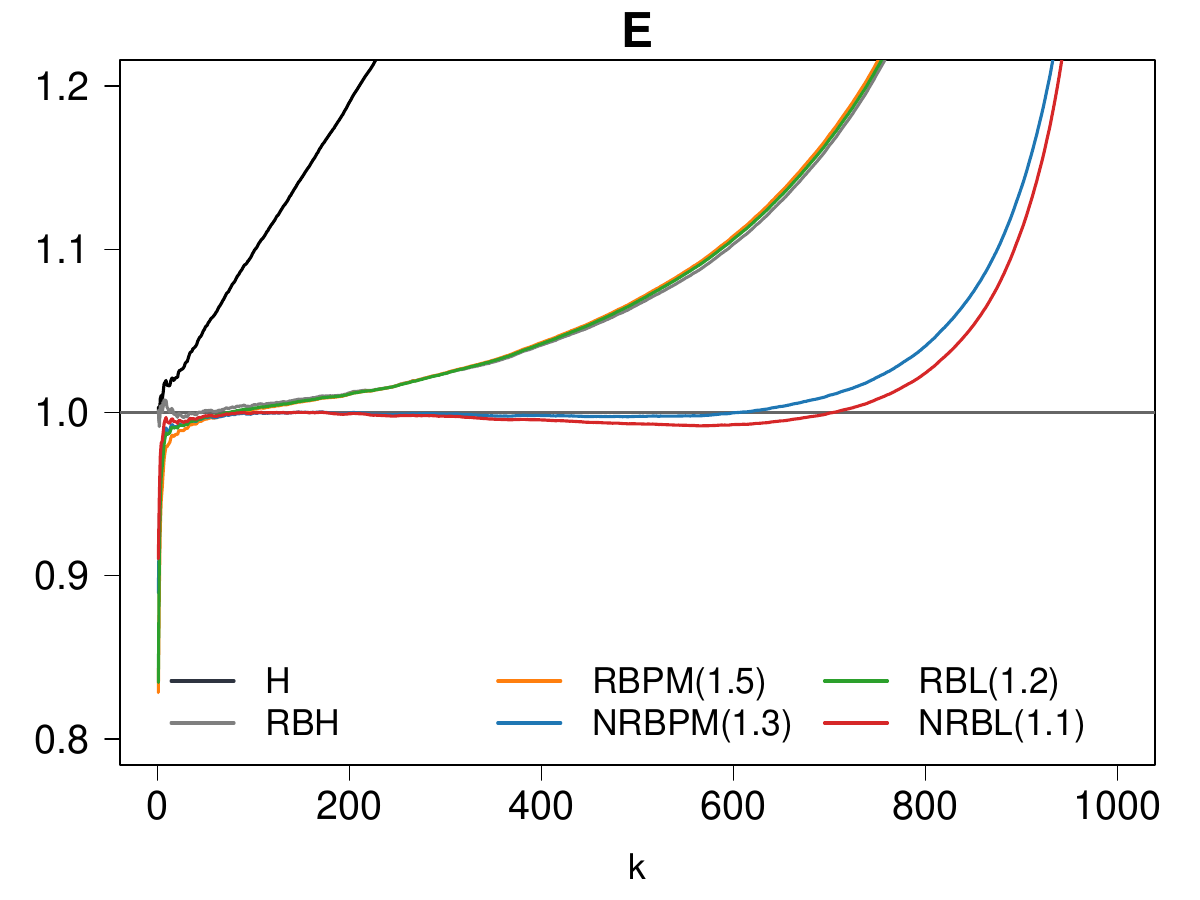}\:\:
\includegraphics[width=.475\textwidth]{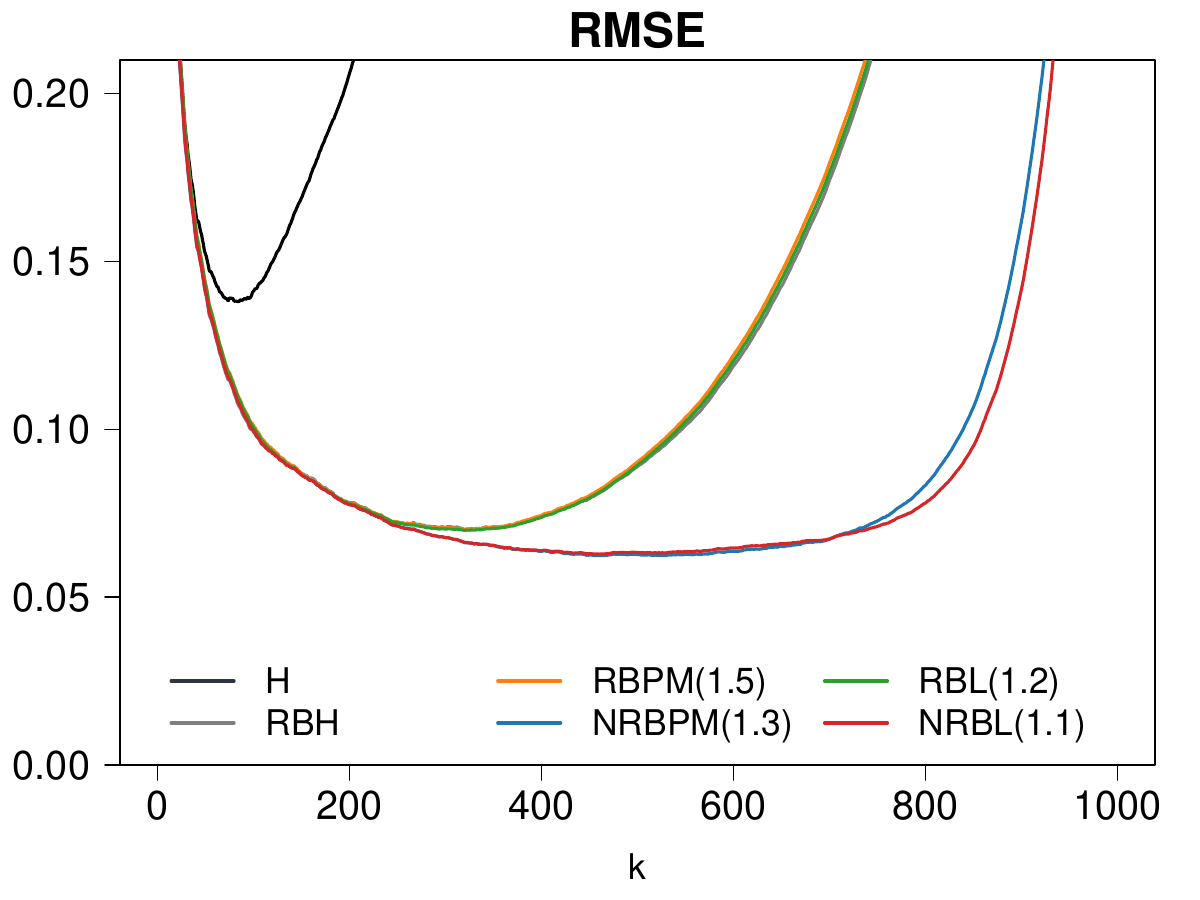}%
\caption{%
Simulated mean values (left) and RMSE (right) of the EVI-estimators
under study for samples of size 1000 from the Burr model, with $\xi=1$ and $\rho=-0.8$ 
}
\label{fig_burr1}
\end{figure}

\begin{figure}[h!]
\includegraphics[width=.475\textwidth]{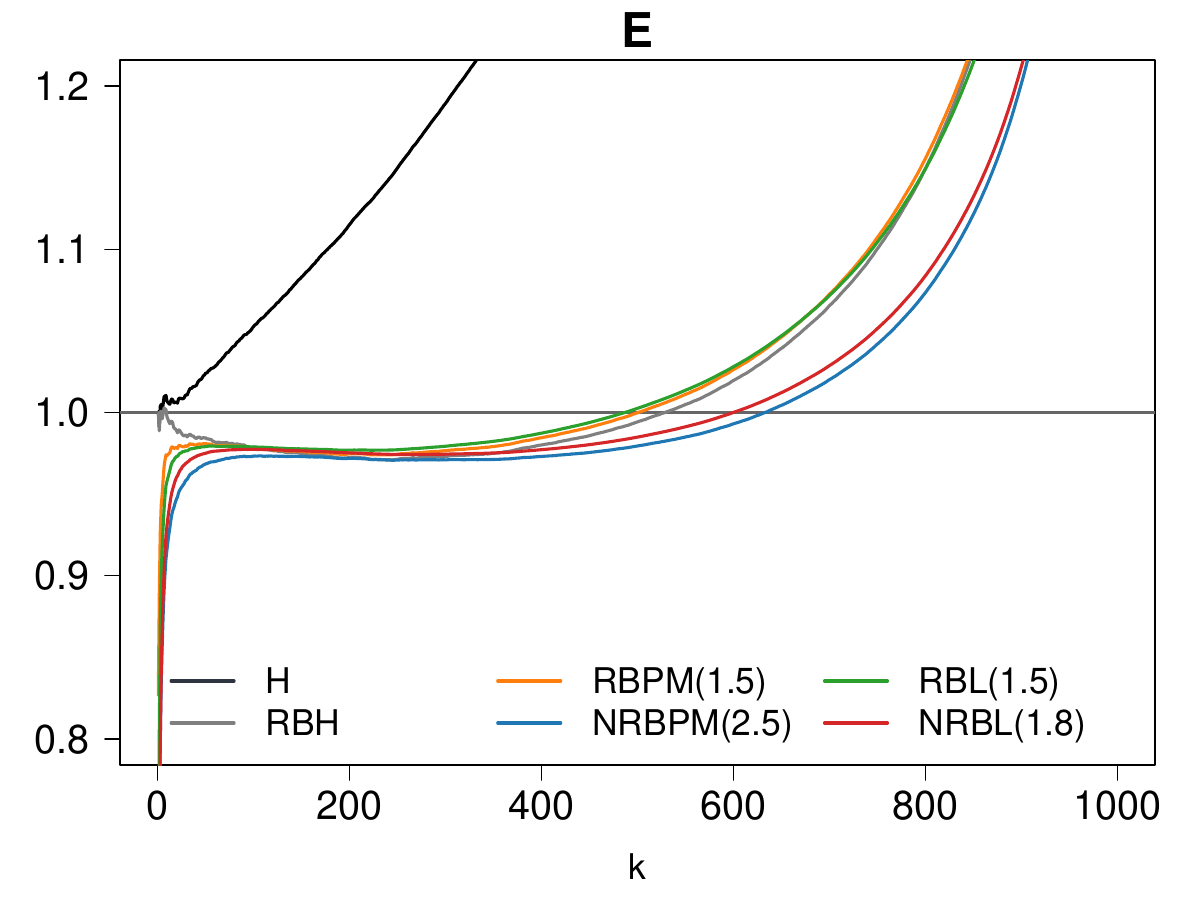}\:\:
\includegraphics[width=.475\textwidth]{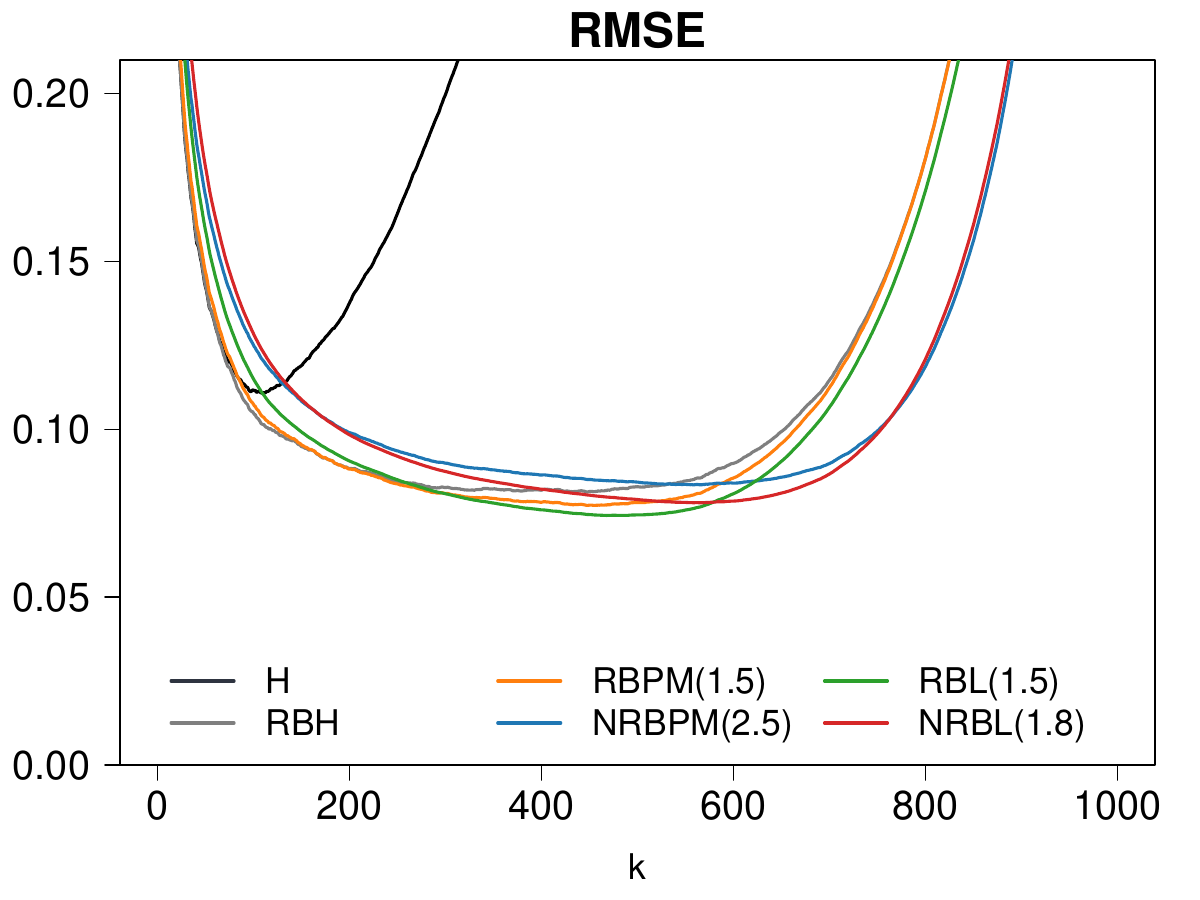}%
\caption{%
Simulated mean values (left) and RMSE (right) of the EVI-estimators
under study for samples of size 1000 from the GP model, with $\xi=1$
}
\label{fig_burr2}%
\end{figure}

\begin{figure}[h!]
\includegraphics[width=.475\textwidth]{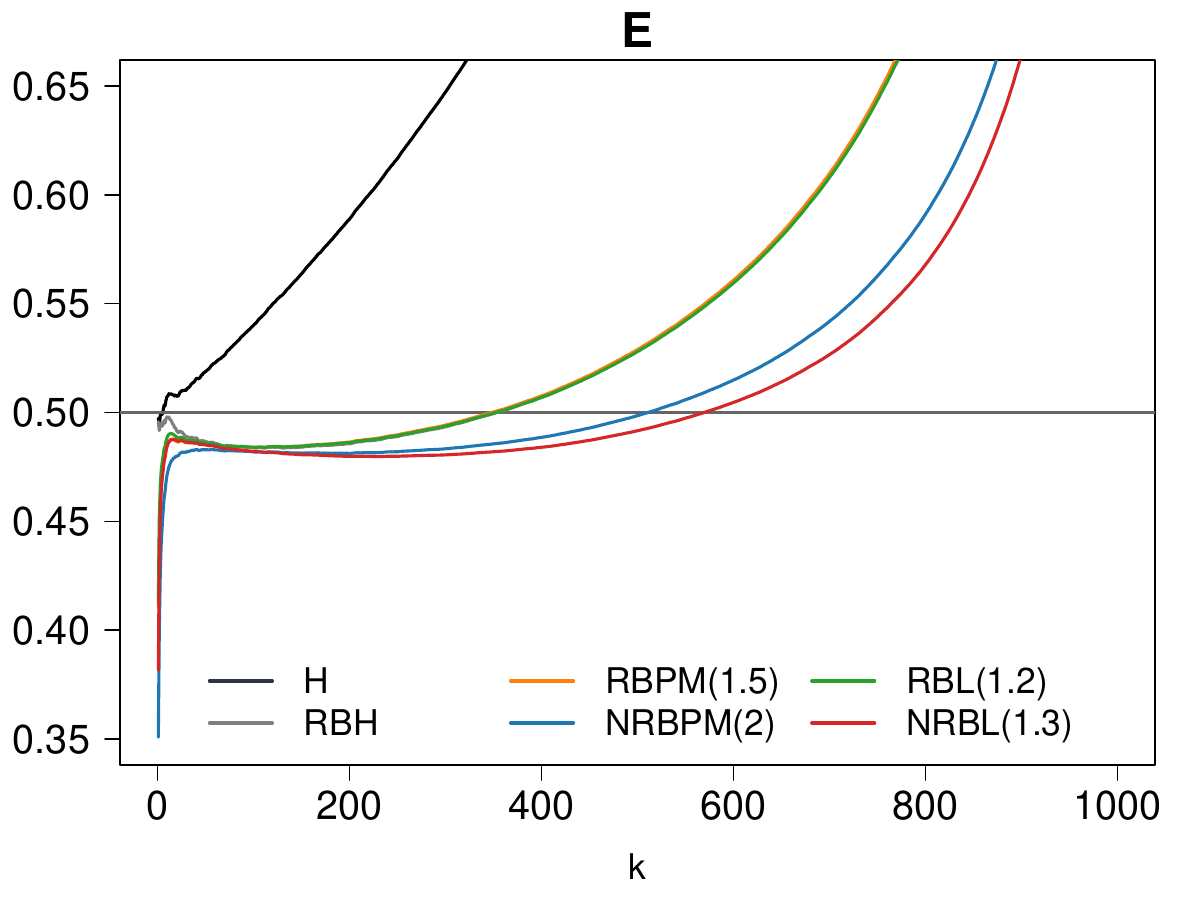}\:\:
\includegraphics[width=.475\textwidth]{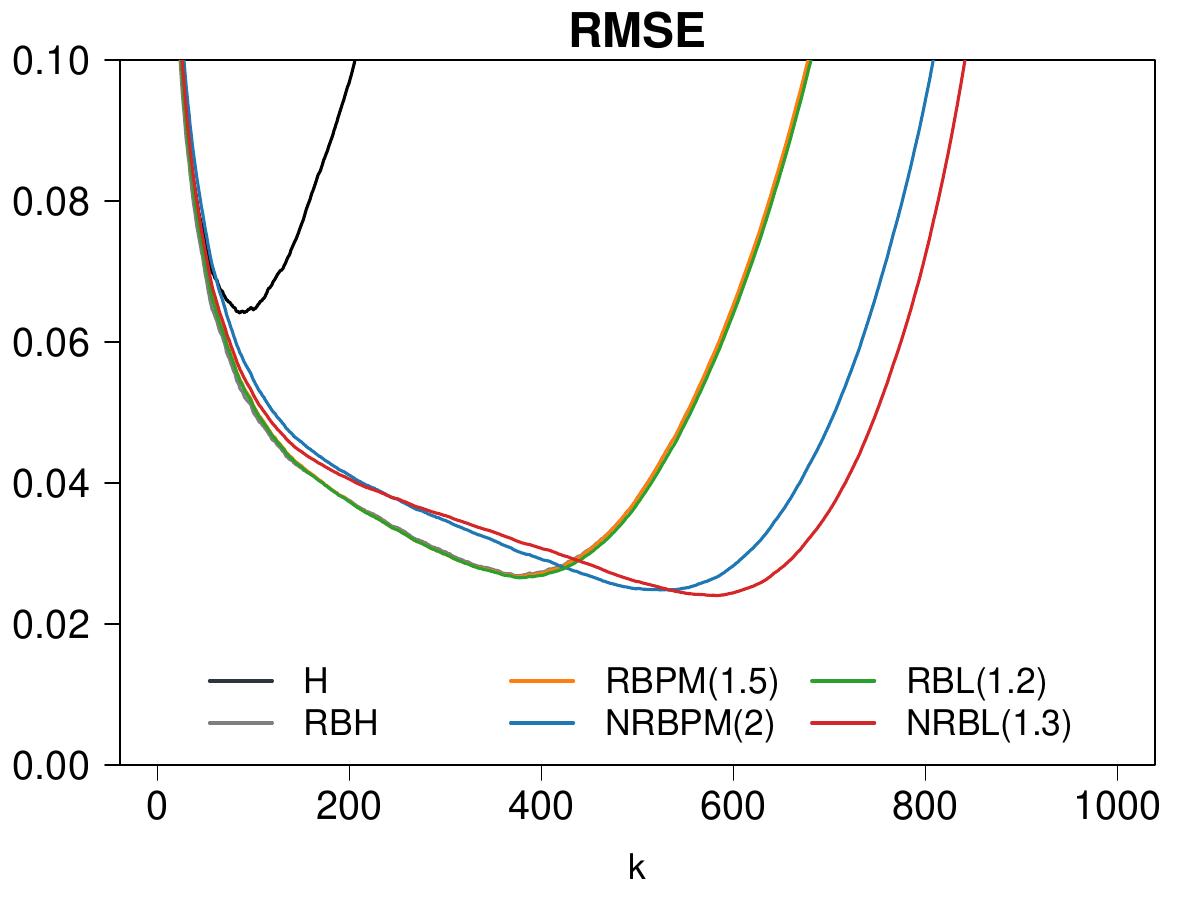}%
\caption{%
Simulated mean values (left) and RMSE (right) of the EVI-estimators
under study for samples of size 1000 from the  Half-$t_2$ model ($\xi=0.5$ and $\rho=-1$)
}
\label{fig_halft2}%
\end{figure}

\subsection{Simulated results at the simulated optimal threshold}%
Tables \ref{tab_frechet}--\ref{tab_ht2} present performance metrics evaluated at the simulated optimal threshold, defined in \eqref{eq:k0}. For each model and sample size, the best achieved values are highlighted in bold.
An analysis of the tables led to the following conclusions:
\begin{itemize}
    \item Across all four models, the Hill and RBH estimators are outperformed by  every general RB/NRB class considered, in both bias and RMSE.
    \item The differences in RMSE between the best-performing RB/NRB class of  estimators and the remaining RB/NRB classes narrow considerably as 
    the sample size $n$ increases, indicating that the choice of tuning parameter 
    $p$  matters most for small-to-moderate sample sizes.
    \item For the Fréchet and GP models, the linear RB estimators, particularly RBL(p)  with $p$ in the [1.2,1.5] range,  most often attain the lowest RMSE. 
    \item For the Burr and Half-$t_2$ models, the pattern reverses. The nonlinear classes, NRBPM and NRBL, tend to achieve the lowest RMSE, particularly from moderate to large sample sizes.
\end{itemize}

\begin{table}[h!]
\centering
\caption{
Simulated mean value and RMSE for the EVI-estimators under study,\ computed at the simulated optimal level, from samples of the Fréchet model with $\xi=0.5$ \label{tab_frechet}}%
\begin{scriptsize}
\begin{tabular}{lrrrrrrr}
  \hline
Estimator & $n=50$ & $n=100$ & $n=250$ & $n=500$ & $n=1000$ & $n=2500$ & $n=5000$ \\ 
  \hline
  \multicolumn{8}{c}{$\mathrm{E}(\widehat{\xi}^{\bullet}(p,\hat k_0^\bullet))$}\\
  H & 0.5734 & 0.5558 & 0.5403 & 0.5331 & 0.5247 & 0.5190 & 0.5130 \\ 
  RBH & 0.4697 & 0.4749 & 0.4894 & 0.4955 & 0.4998 & 0.5015 & 0.5016 \\ 
  RBPM(0.8) & 0.4684 & 0.4754 & 0.4881 & 0.4959 & 0.4995 & 0.5013 & 0.5014 \\ 
  RBPM(1.2) & 0.4743 & 0.4763 & 0.4905 & 0.4957 & 0.4998 & 0.5018 & 0.5017 \\ 
  RBPM(1.5) & 0.4837 & 0.4780 & 0.4924 & 0.4963 & \textbf{0.5001} & 0.5020 & 0.5015 \\ 
  NRBPM(1.6) & 0.4544 & 0.4721 & 0.4865 & 0.4943 & 0.4979 & 0.5007 & 0.5009 \\ 
  NRBPM(2.2) & 0.4591 & 0.4718 & 0.4868 & 0.4937 & 0.4982 & 0.5008 & \textbf{0.5004} \\ 
  NRBPM(2.5) & 0.4667 & 0.4710 & 0.4878 & 0.4938 & 0.4985 & 0.5009 & 0.5007 \\ 
  RBL(0.9) & 0.4687 & 0.4756 & 0.4880 & 0.4960 & 0.4994 & 0.5012 & 0.5014 \\ 
  RBL(1.2) & 0.4827 & 0.4777 & 0.4922 & 0.4964 & 0.5002 & 0.5021 & 0.5016 \\ 
  RBL(1.5) & \textbf{0.4986} & \textbf{0.4909} & 0.5021 & 0.4990 & 0.5017 & 0.5029 & 0.5025 \\ 
  NRBL(1.2) & 0.4544 & 0.4726 & 0.4861 & 0.4941 & 0.4982 & \textbf{0.5006} & 0.5008 \\ 
  NRBL(1.5) & 0.4599 & 0.4738 & 0.4877 & 0.4946 & 0.4988 & 0.5013 & 0.5010 \\ 
  NRBL(2) & 0.4905 & 0.4785 & \textbf{0.5010} & \textbf{0.4996} & 0.5037 & 0.5038 & 0.5033 \\ 
  \hline
  \multicolumn{8}{c}{$\mathrm{RMSE}(\widehat{\xi}^{\bullet}(p,\hat k_0^\bullet))$}\\  
  H & 0.1384 & 0.1055 & 0.0756 & 0.0594 & 0.0455 & 0.0336 & 0.0259 \\ 
  RBH & 0.1222 & 0.1035 & 0.0702 & 0.0554 & 0.0405 & 0.0274 & 0.0199 \\ 
  RBPM(0.8) & 0.1232 & 0.1048 & 0.0712 & 0.0559 & 0.0407 & 0.0275 & 0.0200 \\ 
  RBPM(1.2) & 0.1207 & 0.1024 & 0.0696 & 0.0551 & 0.0405 & 0.0274 & 0.0199 \\ 
  RBPM(1.5) & 0.1191 & 0.1016 & 0.0693 & 0.0551 & 0.0408 & 0.0276 & 0.0200 \\ 
  NRBPM(1.6) & 0.1372 & 0.1154 & 0.0762 & 0.0587 & 0.0424 & 0.0284 & 0.0205 \\ 
  NRBPM(2.2) & 0.1354 & 0.1157 & 0.0765 & 0.0596 & 0.0439 & 0.0295 & 0.0213 \\ 
  NRBPM(2.5) & 0.1331 & 0.1160 & 0.0769 & 0.0604 & 0.0448 & 0.0303 & 0.0218 \\ 
  RBL(0.9) & 0.1238 & 0.1055 & 0.0715 & 0.0562 & 0.0409 & 0.0277 & 0.0201 \\ 
  RBL(1.2) & 0.1183 & 0.1010 & 0.0688 & 0.0546 & \textbf{0.0403} & \textbf{0.0273} & \textbf{0.0198} \\ 
  RBL(1.5) & \textbf{0.1144} & \textbf{0.0984} & \textbf{0.0675} & \textbf{0.0543} & 0.0406 & 0.0276 & 0.0199 \\ 
  NRBL(1.2) & 0.1376 & 0.1153 & 0.0757 & 0.0582 & 0.0417 & 0.0279 & 0.0202 \\ 
  NRBL(1.5) & 0.1334 & 0.1140 & 0.0745 & 0.0579 & 0.0420 & 0.0282 & 0.0204 \\ 
  NRBL(2) & 0.1229 & 0.1112 & 0.0717 & 0.0576 & 0.0425 & 0.0290 & 0.0210 \\ 
   \hline
\end{tabular}
\end{scriptsize}
\end{table}

\begin{table}[h!]
\centering
\caption{
Simulated mean value and RMSE for the EVI-estimators under study,\ computed at the simulated optimal level, from samples of the Burr model with $\xi=1$ and $\rho=-0.8$ ($\delta=1$) \label{tab_burr}}%
\begin{scriptsize}
\begin{tabular}{lrrrrrrr}
  \hline
Estimator & $n=50$ & $n=100$ & $n=250$ & $n=500$ & $n=1000$ & $n=2500$ & $n=5000$ \\ 
  \hline
  \multicolumn{8}{c}{$\mathrm{E}(\widehat{\xi}^{\bullet}(p,\hat k_0^\bullet))$}\\
  H & 1.2259 & 1.1711 & 1.1178 & 1.0911 & 1.0848 & 1.0622 & 1.0516 \\ 
  RBH & 1.0825 & 1.0587 & 1.0419 & 1.0358 & 1.0267 & 1.0197 & 1.0176 \\ 
  RBPM(0.8) & 1.0883 & 1.0604 & 1.0407 & 1.0351 & 1.0261 & 1.0193 & 1.0171 \\ 
  RBPM(1.2) & 1.0769 & 1.0569 & 1.0416 & 1.0361 & 1.0271 & 1.0201 & 1.0175 \\ 
  RBPM(1.5) & 1.0900 & 1.0594 & 1.0409 & 1.0362 & 1.0273 & 1.0205 & 1.0168 \\ 
  NRBPM(0.8) & \textbf{1.0282} & \textbf{0.9981} & 0.9846 & 0.9789 & 0.9818 & 0.9901 & 0.9887 \\ 
  NRBPM(1.3) & 1.0528 & 1.0113 & 1.0024 & \textbf{0.9986} & \textbf{0.9977} & 1.0022 & 1.0009 \\ 
  NRBPM(1.5) & 1.0586 & 1.0173 & 1.0083 & 1.0062 & 1.0084 & 1.0097 & 1.0072 \\ 
  RBL(0.8) & 1.0984 & 1.0653 & 1.0422 & 1.0349 & 1.0260 & 1.0221 & 1.0167 \\ 
  RBL(1.2) & 1.0914 & 1.0597 & 1.0423 & 1.0365 & 1.0270 & 1.0203 & 1.0178 \\ 
  RBL(1.5) & 1.0913 & 1.0670 & 1.0485 & 1.0355 & 1.0284 & 1.0219 & 1.0182 \\ 
  NRBL(0.8) & 1.0323 & 0.9919 & 0.9818 & 0.9748 & 0.9806 & 0.9854 & 0.9864 \\ 
  NRBL(1.1) & 1.0496 & 1.0082 & \textbf{0.9994} & 0.9957 & 0.9938 & \textbf{0.9996} & \textbf{1.0002} \\ 
  NRBL(1.3) & 1.0626 & 1.0254 & 1.0178 & 1.0110 & 1.0107 & 1.0126 & 1.0105 \\
  \hline%
  \multicolumn{8}{c}{$\mathrm{RMSE}(\widehat{\xi}^{\bullet}(p,\hat k_0^\bullet))$}\\  
  H & 0.4122 & 0.3163 & 0.2247 & 0.1758 & 0.1380 & 0.1045 & 0.0841 \\ 
  RBH & 0.2941 & 0.2022 & 0.1363 & 0.0972 & 0.0702 & 0.0464 & 0.0362 \\ 
  RBPM(0.8) & 0.3021 & 0.2055 & 0.1384 & 0.0984 & 0.0708 & 0.0465 & 0.0361 \\ 
  RBPM(1.2) & 0.2881 & 0.2000 & 0.1349 & 0.0965 & 0.0699 & 0.0465 & 0.0364 \\ 
  RBPM(1.5) & 0.2814 & 0.1980 & 0.1340 & 0.0963 & 0.0701 & 0.0471 & 0.0369 \\ 
  NRBPM(0.8) & 0.2997 & 0.2057 & 0.1386 & 0.0977 & 0.0685 & 0.0426 & 0.0321 \\ 
  NRBPM(1.3) & 0.2824 & 0.1980 & 0.1313 & 0.0912 & 0.0624 & 0.0339 & \textbf{0.0248} \\ 
  NRBPM(1.5) & 0.2779 & 0.1963 & 0.1303 & 0.0905 & \textbf{0.0620} & 0.0357 & 0.0271 \\ 
  RBL(0.8) & 0.3206 & 0.2148 & 0.1444 & 0.1019 & 0.0731 & 0.0478 & 0.0369 \\ 
  RBL(1.2) & 0.2827 & 0.1978 & 0.1336 & 0.0959 & 0.0699 & 0.0468 & 0.0367 \\ 
  RBL(1.5) & \textbf{0.2731} & 0.1966 & 0.1336 & 0.0969 & 0.0714 & 0.0488 & 0.0382 \\ 
  NRBL(0.8) & 0.3154 & 0.2141 & 0.1456 & 0.1038 & 0.0727 & 0.0455 & 0.0349 \\ 
  NRBL(1.1) & 0.2842 & 0.1986 & 0.1318 & 0.0916 & 0.0628 & \textbf{0.0335} & 0.0249 \\ 
  NRBL(1.3) & 0.2750 & \textbf{0.1942} & \textbf{0.1290} & \textbf{0.0898} & 0.0621 & 0.0375 & 0.0289 \\
   \hline
\end{tabular}
\end{scriptsize}
\end{table}

\begin{table}[h!]
\centering
\caption{
Simulated mean value and RMSE for the EVI-estimators under study,\ computed at the simulated optimal level, from samples of the GP model with $\xi=1$ ($\rho=-1$) \label{tab_gp}}%
\begin{scriptsize}
\begin{tabular}{lrrrrrrr}
  \hline
Estimator & $n=50$ & $n=100$ & $n=250$ & $n=500$ & $n=1000$ & $n=2500$ & $n=5000$ \\ 
  \hline
  \multicolumn{8}{c}{$\mathrm{E}(\widehat{\xi}^{\bullet}(p,\hat k_0^\bullet))$}\\  
  H & 1.1551 & 1.1268 & 1.1064 & 1.0762 & 1.0581 & 1.0464 & 1.0346 \\ 
  RBH & 1.0486 & 1.0052 & 0.9887 & 0.9858 & 0.9855 & 1.0019 & 1.0021 \\ 
  RBPM(1.2) & 1.0476 & 1.0061 & 1.0061 & 0.9876 & 0.9878 & 0.9991 & 1.0008 \\ 
  RBPM(1.5) & 1.0453 & 1.0171 & 1.0081 & 0.9920 & 0.9915 & 1.0020 & 1.0026 \\ 
  RBPM(1.8) & 1.0419 & 1.0272 & 1.0092 & 0.9965 & 0.9956 & 1.0021 & 1.0012 \\ 
  NRBPM(2) & \textbf{0.9939} & 0.9908 & 0.9723 & 0.9686 & 0.9661 & 0.9988 & 0.9990 \\ 
  NRBPM(2.5) & 1.0243 & 1.0062 & 0.9939 & 0.9871 & 0.9862 & 1.0010 & 0.9992 \\ 
  NRBPM(2.8) & 1.0254 & 1.0100 & \textbf{1.0002} & 0.9954 & 0.9924 & 1.0019 & 1.0003 \\ 
  RBL(1.2) & 1.0452 & 1.0194 & 1.0068 & 0.9910 & 0.9903 & \textbf{1.0002} & 1.0020 \\ 
  RBL(1.5) & 1.0459 & 1.0285 & 1.0129 & \textbf{1.0028} & 0.9967 & 1.0031 & 1.0022 \\ 
  RBL(2) & 1.0651 & 1.0437 & 1.0286 & 1.0174 & 1.0108 & 1.0073 & 1.0046 \\ 
  NRBL(1.5) & 1.0143 & \textbf{0.9952} & 0.9772 & 0.9721 & 0.9698 & 0.9994 & \textbf{0.9997} \\ 
  NRBL(1.8) & 1.0327 & 1.0142 & 1.0017 & 0.9970 & 0.9935 & 1.0010 & 1.0010 \\ 
  NRBL(2) & 1.0440 & 1.0305 & 1.0145 & 1.0044 & \textbf{1.0002} & 1.0031 & 1.0027 \\
  \hline%
  \multicolumn{8}{c}{$\mathrm{RMSE}(\widehat{\xi}^{\bullet}(p,\hat k_0^\bullet))$}\\  
  H & 0.3525 & 0.2658 & 0.1860 & 0.1441 & 0.1107 & 0.0832 & 0.0655 \\ 
  RBH & 0.2503 & 0.1895 & 0.1376 & 0.1057 & 0.0813 & 0.0386 & 0.0262 \\ 
  RBPM(1.2) & 0.2462 & 0.1872 & 0.1359 & 0.1040 & 0.0793 & 0.0381 & 0.0261 \\ 
  RBPM(1.5) & 0.2420 & 0.1850 & 0.1336 & 0.1024 & 0.0773 & 0.0377 & 0.0264 \\ 
  RBPM(1.8) & 0.2394 & 0.1839 & 0.1328 & 0.1020 & 0.0763 & 0.0379 & 0.0268 \\ 
  NRBPM(2) & 0.2483 & 0.1960 & 0.1452 & 0.1130 & 0.0875 & 0.0417 & 0.0287 \\ 
  NRBPM(2.5) & 0.2434 & 0.1916 & 0.1427 & 0.1103 & 0.0835 & 0.0406 & 0.0284 \\ 
  NRBPM(2.8) & 0.2409 & 0.1907 & 0.1423 & 0.1098 & 0.0824 & 0.0411 & 0.0290 \\ 
  RBL(1.2) & 0.2421 & 0.1844 & 0.1333 & 0.1021 & 0.0773 & \textbf{0.0376} & \textbf{0.0261} \\ 
  RBL(1.5) & 0.2376 & \textbf{0.1812} & \textbf{0.1304} & \textbf{0.0996} & \textbf{0.0743} & 0.0377 & 0.0268 \\ 
  RBL(2) & \textbf{0.2366} & 0.1830 & 0.1324 & 0.1002 & 0.0743 & 0.0408 & 0.0298 \\ 
  NRBL(1.5) & 0.2464 & 0.1920 & 0.1413 & 0.1100 & 0.0848 & 0.0406 & 0.0279 \\ 
  NRBL(1.8) & 0.2401 & 0.1865 & 0.1363 & 0.1043 & 0.0781 & 0.0393 & 0.0277 \\ 
  NRBL(2) & 0.2385 & 0.1858 & 0.1355 & 0.1030 & 0.0766 & 0.0400 & 0.0287 \\ 
   \hline
\end{tabular}
\end{scriptsize}
\end{table}

\begin{table}[h!]
\centering
\caption{
Simulated mean value and RMSE for the EVI-estimators under study,\ computed at the simulated optimal level, from samples of the Half-$t_2$ model ($\xi=0.5$ and $\rho=-1$) \label{tab_ht2}}%
\begin{scriptsize}
\begin{tabular}{lrrrrrrr}
  \hline
Estimator & $n=50$ & $n=100$ & $n=250$ & $n=500$ & $n=1000$ & $n=2500$ & $n=5000$ \\ 
  \hline
  \multicolumn{8}{c}{$\mathrm{E}(\widehat{\xi}^{\bullet}(p,\hat k_0^\bullet))$}\\  
  H & 0.6003 & 0.5647 & 0.5549 & 0.5469 & 0.5340 & 0.5286 & 0.5228 \\ 
  RBH & 0.5327 & 0.5220 & 0.5104 & 0.5039 & 0.5032 & 0.5017 & 0.5011 \\ 
  RBPM(1.2) & 0.5302 & 0.5211 & 0.5102 & 0.5040 & 0.5035 & 0.5017 & 0.5015 \\ 
  RBPM(1.5) & 0.5264 & 0.5197 & 0.5134 & 0.5061 & 0.5039 & 0.5022 & 0.5019 \\ 
  RBPM(1.8) & 0.5335 & 0.5181 & 0.5136 & 0.5076 & 0.5041 & 0.5026 & 0.5019 \\ 
  NRBPM(1.5) & 0.5137 & 0.4984 & 0.4935 & 0.4974 & 0.4983 & 0.5008 & 0.4997 \\ 
  NRBPM(2) & 0.5158 & 0.5062 & 0.5028 & 0.5018 & 0.5032 & 0.5008 & 0.5004 \\ 
  NRBPM(2.5) & 0.5144 & 0.5099 & 0.5080 & 0.5079 & 0.5038 & 0.5024 & 0.5009 \\ 
  RBL(0.8) & 0.5408 & 0.5258 & 0.5116 & 0.5041 & 0.5023 & 0.5013 & 0.5006 \\ 
  RBL(1.2) & 0.5270 & 0.5193 & 0.5128 & 0.5055 & 0.5037 & 0.5016 & 0.5014 \\ 
  RBL(1.5) & 0.5394 & 0.5223 & 0.5128 & 0.5080 & 0.5042 & 0.5033 & 0.5017 \\ 
  NRBL(1.1) & \textbf{0.5102} & 0.4911 & 0.4856 & 0.4917 & \textbf{0.4988} & \textbf{0.5000} & 0.5004 \\ 
  NRBL(1.3) & 0.5145 & \textbf{0.5008} & \textbf{0.5013} & \textbf{0.4996} & 0.5018 & 0.5008 & \textbf{0.5001} \\ 
  NRBL(1.8) & 0.5248 & 0.5157 & 0.5100 & 0.5065 & 0.5041 & 0.5025 & 0.5012 \\ 
  \hline%
  \multicolumn{8}{c}{$\mathrm{RMSE}(\widehat{\xi}^{\bullet}(p,\hat k_0^\bullet))$}\\  
  H & 0.2021 & 0.1514 & 0.1069 & 0.0816 & 0.0641 & 0.0479 & 0.0379 \\ 
  RBH & 0.1384 & 0.0901 & 0.0565 & 0.0399 & 0.0268 & 0.0173 & 0.0125 \\ 
  RBPM(1.2) & 0.1354 & 0.0887 & 0.0558 & 0.0395 & 0.0266 & 0.0172 & 0.0125 \\ 
  RBPM(1.5) & 0.1325 & 0.0875 & 0.0554 & 0.0394 & 0.0267 & 0.0174 & 0.0127 \\ 
  RBPM(1.8) & 0.1303 & 0.0871 & 0.0555 & 0.0396 & 0.0270 & 0.0178 & 0.0130 \\ 
  NRBPM(1.5) & 0.1302 & 0.0848 & 0.0529 & 0.0374 & 0.0243 & 0.0150 & \textbf{0.0111} \\ 
  NRBPM(2) & 0.1268 & 0.0839 & 0.0523 & 0.0370 & 0.0248 & 0.0158 & 0.0116 \\ 
  NRBPM(2.5) & \textbf{0.1252} & 0.0847 & 0.0539 & 0.0384 & 0.0261 & 0.0170 & 0.0125 \\ 
  RBL(0.8) & 0.1512 & 0.0973 & 0.0607 & 0.0425 & 0.0286 & 0.0183 & 0.0132 \\ 
  RBL(1.2) & 0.1329 & 0.0874 & 0.0551 & 0.0392 & 0.0265 & 0.0172 & 0.0125 \\ 
  RBL(1.5) & 0.1286 & 0.0864 & 0.0552 & 0.0395 & 0.0272 & 0.0179 & 0.0130 \\ 
  NRBL(1.1) & 0.1333 & 0.0867 & 0.0552 & 0.0398 & 0.0254 & \textbf{0.0150} & 0.0112 \\ 
  NRBL(1.3) & 0.1286 & \textbf{0.0837} & \textbf{0.0521} & \textbf{0.0368} & \textbf{0.0240} & 0.0150 & 0.0111 \\ 
  NRBL(1.8) & 0.1254 & 0.0843 & 0.0536 & 0.0382 & 0.0262 & 0.0172 & 0.0125 \\ 
   \hline
\end{tabular}
\end{scriptsize}
\end{table}

\section{Hydrological case study}\label{sec:apps}

This section illustrates the performance of the general classes of RB EVI-estimators through a hydrological case study. The data set consists of daily river flow measurements (in m$^3$/s) collected at a monitoring station of the Whiteadder river, UK, over a multi-year period. The data, covering the period from 1990 to 2005, was collected from \url{https://nrfa.ceh.ac.uk/data/station/meanflow/21022} and contains $n=5844$ observations. Preliminary exploratory analysis, including time series plots and boxplots, presented in Figure \ref{fig:river_des} indicates the presence of extreme events and strong right-skewness (see Table \ref{tab:rio}) and heavy-tailed behaviour, making it suitable for {the type of extreme value analysis in this article}.

\begin{table}[h!]
	\caption{Summary statistics of the Whiteadder river flow}
    \centering
	\begin{tabular}{lllllll}
		$n$ & min & max & median & mean & skewness &  kurtosis \\ \hline
		5844 &  0.707 &  219.700& 3.760 & 6.573 & 7.999 & 110.835
	\end{tabular} \label{tab:rio}
\end{table}

\begin{figure}[h!]
    \begin{tabular}{ll}
    \includegraphics[width=\textwidth]{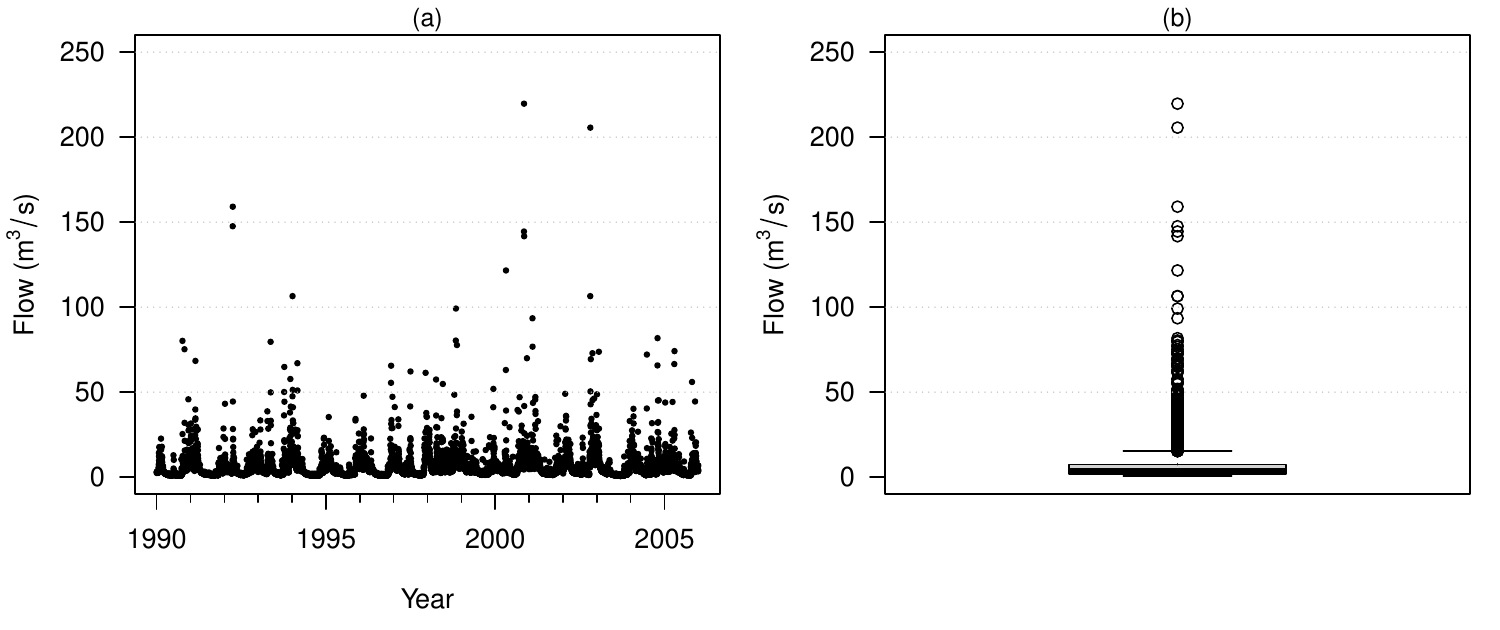}
\vspace{-0.35cm}
    \end{tabular}
	\caption{Daily time series (a) and boxplot (b) of the Whiteadder river flow between 1990 and 2005}
	\label{fig:river_des}
\end{figure}

The EVI is estimated using the classical Hill estimator and the general classes of power mean and $L_p$ estimators, together with their corresponding RB and nonlinear RB versions. Figure~\ref{fig:riverpaths} displays the sample paths of selected estimators as a function of the number of upper order statistics $k$. The classical Hill estimator exhibits considerable variability and a pronounced sensitivity to the choice of $k$. In contrast, RB estimators show extended stability regions, allowing for a more reliable identification of appropriate $k$ values. In particular, a suitable choice of RB power mean and $L_p$ estimators produces smoother trajectories and reduced variability across a wide range of thresholds. 

\begin{figure}[h!]
\center{\includegraphics[width=0.75\textwidth]{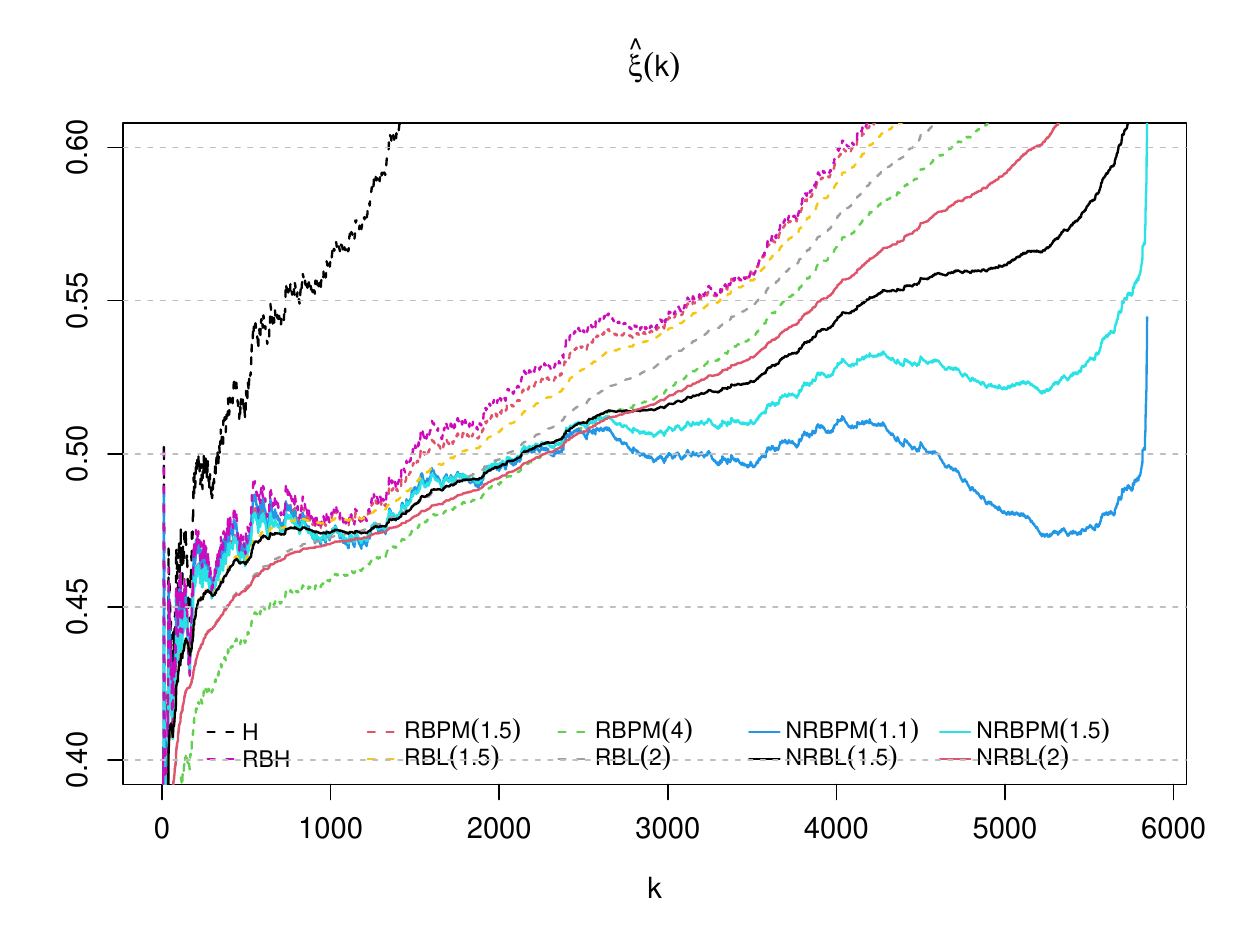}} \vspace{-0.5cm}
 \caption{Sample paths of selected EVI-estimators for the river flow data set}
\label{fig:riverpaths}
\end{figure}

Table \ref{tab:river_estimates} presents the EVI estimates obtained using the sample-path stability threshold selection procedure described by \citet{caeiro2016threshold}. The procedure selects the number \(k\) of upper order statistics by identifying a stable region in the sequence of EVI estimates. First, the minimum number of decimal places required to distinguish the estimates is determined, and the longest run of equal consecutive rounded estimates is identified. Within this region, the estimates are then considered with two additional decimal places and their modal value is computed. The selected value \(k^\ast\) is the largest \(k\) associated with this modal value, and the corresponding estimate \(\widehat{\xi}(k^\ast)\) is retained as the final EVI estimate. The estimated standard errors, $\widehat{\rm se}$, depend on the class of EVI-estimators under study, \{H, RBH, RBPM, NRBPM, RBL, NRBL\}, and on the value of $p$. They are obtained from the corresponding asymptotic variance expressions in \eqref{eq:var_H}, \eqref{eq:var_PM} and \eqref{eq:var_L}, and computed as 
\[\widehat{\rm se} := \hat \xi (k^\ast) \frac{\varphi(p)}{\sqrt{k^\ast}},\] 
where
\begin{equation*}
   \varphi(p):= \left\{
   \begin{array}{cc}
   1, &  p=1 \wedge \{\rm{H, RBH, RBPM, NRBPM, RBL, NRBL}\}\\
\frac{1}{p} \sqrt{\frac{\Gamma(2p+1)}{\Gamma^2(p+1)}-1}, & p>-\dfrac12,\,\, p\neq \{0, 1\} \wedge \{\rm{RBPM, NRBPM}\}\\
\sqrt{\frac{\Gamma(2p-1}{\Gamma^2(p)}}, & p >\dfrac12, \,\, p\neq 1 \wedge \{\rm{RBL, NRBL}\}.
\end{array}
\right.
\end{equation*}

\begin{table}[h!]
\centering
\caption{
Optimal threshold, \(k^\ast\), EVI estimates and associated standard errors, obtained using the sample-path stability threshold selection procedure proposed by {\citet{caeiro2016threshold}} \label{tab:river_estimates}}%
\begin{tabular}{p{3cm} p{1cm} p{1.25cm} p{1.25cm}}
\hline\noalign{\vskip 1pt}
Estimator & \(k^\ast\) & $\hat \xi (k^\ast)$ & $\widehat{\rm se}$\\[1pt]
\hline
H & 942 & 0.5545 & 0.0181 \\
RBH & 3028 & 0.5450 & 0.0099 \\
RBPM(1.5) & 3043 & 0.5446 & {0.0102} \\
RBPM(4) & 2792 & 0.5150 & {0.0202}\\ 
RBL(1.5) & 3140 & 0.5450 & {0.0110}\\  
RBL(2) & 2793 & 0.5250 & {0.0140}\\  
NRBPM(1.1) & 1329 & 0.4750 & {0.0130
}\\ 
NRBPM(1.5) & 3589 & 0.5148 & {0.0089}\\
NRBL(1.5) & 5101 & 0.5650 & {0.0089}\\ 
NRBL(2) & 2849 & 0.5149 & {0.0136}\\ 
\hline 
\end{tabular}
\end{table} 


The stability based threshold selection procedure leads to different values of $k^\ast$ across estimators, reflecting their distinct bias/variance behaviour. The RB and nonlinear RB estimators generally exhibit lower estimated standard errors than the Hill estimator. 
In particular, NRBPM(1.1) combines a moderate selected threshold with a stable sample path.

\section{Conclusions and future work}\label{sec:conclusions}

In this paper, we have presented a comparative study of general classes of linear and nonlinear RB EVI-estimators, encompassing power mean-type and $L_p$-type methodologies within a unified third-order framework. The theoretical and simulation results show that the relative performance of linear and nonlinear RB estimators depends on the underlying tail and second-order structure. The proposed NRBPM and NRBL estimators perform particularly well for the Burr and Half-$t_2$ models, whereas linear RB estimators remain highly competitive for the Fr\'echet and GP models. Differences among the best performing procedures decrease as the sample size increases, suggesting that the influence of the tuning parameter $p$ is more pronounced in small and moderate samples and diminishes for larger sample sizes. The empirical application further supports the practical relevance of the nonlinear RB classes and confirms that different choices of $p$ may lead to  different sample path behaviour. These findings motivate future research on data-driven selection of $p$, jointly with an asymptotically valid choice of the intermediate sequence $k=k_n$, possibly through criteria combining sample-path stability with bias and variance considerations.

\section*{Acknowledgements}
\noindent Research partially supported by national funds through FCT – Fundação para a Ciência e a Tecnologia, I.P., under CEAUL Research Unit, UID/00006/2025 (\url{https://doi.org/10.54499/UID/00006/2025}), and by the European Union---NextGenerationEU through the project, UID/PRR/00006/2025 (\url{https://doi.org/10.54499/UID/PRR/00006/2025}), 
Center for Mathematics and Applications / NOVA Math research unit, 
UID/00297/2025 (https://doi.org/10.54499/ UID/00297/2025) and UID/PRR/00297/2025 (\url{https://doi.org/10.54499/UID/PRR/00297/2025)}.

\bibliographystyle{elsarticle-harv}
\bibliography{referenceRB}

\end{document}